\documentclass[journal,twoside,web,letter]{ieeecolor_noruler}
\usepackage{generic}
\usepackage{cite}
\usepackage{amsmath,amssymb,amsfonts}
\usepackage{algorithmic}
\usepackage{graphicx}
\usepackage{algorithm,algorithmic}
\usepackage{amssymb}
\usepackage{hyperref}
\hypersetup{hidelinks=true}
\usepackage{textcomp}
\usepackage{wrapfig}

\newtheorem{corol}{Corollary}

\newtheorem{lemma}{Lemma}
\newtheorem{remark}{Remark}
\newtheorem{thm}{Theorem}

\DeclareMathOperator{\rank}{rank}

\DeclareMathOperator{\row}{row}

\def\BibTeX{{\rm B\kern-.05em{\sc i\kern-.025em b}\kern-.08em
		T\kern-.1667em\lower.7ex\hbox{E}\kern-.125emX}}
\begin{document}
		\def\ZZ{{\mathbb Z}}
		\def\RR{{\Bbb R}}
		\def\NN{{\mathbb N}}
		\def\CC{{\mathbb C}}
		
	\title{Existence and Design of Functional Observers for Time-Delay Systems with Delayed Output Measurements}
	\author{Hieu Trinh, Phan Thanh Nam and Tyrone Fernando
				\thanks{H.~Trinh is with the School of Engineering, Deakin University, Waurn Ponds, 75 Pigdons Road, Geelong, Australia. (email:  hieu.trinh@deakin.edu.au)}
				\thanks{P. T.~Nam is with the Department of Mathematics, Quy Nhon University, Vietnam. (email:  phanthanhnam@qnu.edu.vn)}
		\thanks{T.~Fernando is with the Department of Electrical, Electronic and Computer Engineering, University of Western Australia (UWA), 35 Stirling Highway, Crawley, WA 6009, Australia. (email:  tyrone.fernando@uwa.edu.au)}	
}
	
	\maketitle
\begin{abstract}
	This paper investigates the problem of functional state estimation for
	linear time-delay systems in which the delay affecting the state evolution
	differs from the delay affecting the output measurements. While existing
	observer designs typically assume instantaneous output availability,
	practical systems often exhibit measurement delays that are distinct
	from and not aligned with the intrinsic state delay. We explicitly
	distinguish between the state delay $\tau$ and the measurement delay $h$
	and address the problem of estimating a desired functional
	$z(t)=Fx(t)$ under such mismatched delay conditions.
	
	Three functional observer structures are proposed to accommodate
	different delay configurations, each capable of realizing functional
	observers of different orders. This flexibility is important since
	a functional observer whose order equals the number of estimated
	functionals may not always exist. For each structure, algebraic existence conditions are established together with
	constructive synthesis procedures. A functional augmentation framework
	is developed to derive verifiable rank-based conditions for observers
	of various orders. In addition, the notion of \emph{generalized
		functionals}, defined over an \emph{augmented delayed state vector}, is
	introduced to provide greater flexibility in satisfying observer
	existence conditions and facilitating systematic design. Numerical
	examples illustrate the proposed theory.
\end{abstract}
	
	
\section{Introduction}

Time-delay systems arise naturally in many engineering, biological,
economic, and physical processes in which the current rate of change
of the state depends not only on its present value but also on its
past history \cite{Richard2003,Niculescu2001,Hale1993,Michiels2007}.
Such delays typically originate from transport phenomena in chemical
processes, communication constraints in networked control systems,
transmission and coordination delays in power systems, regenerative
effects in mechanical systems such as machining processes, and
maturation effects in biological models.
Unlike classical finite-dimensional systems, time-delay systems are
inherently infinite-dimensional since their future evolution depends
on functions defined over an interval of past time.

The presence of delay fundamentally alters system properties such as
stability, controllability, observability, and stabilizability.
Classical quadratic Lyapunov functions are generally insufficient to
capture the behaviour of delayed systems, and stability analysis
typically requires Lyapunov--Krasovskii functionals or related
infinite-dimensional techniques.
Similarly, observer and controller design must explicitly account for
the delayed dynamics to ensure well-posedness and asymptotic
performance.
Stabilizing controllers for time-delay systems are frequently
formulated in a state-feedback form, where the control input depends
linearly on the system state.
In particular, many designs adopt a functional state-feedback
structure of the form
\(
u(t)=Fx(t),
\)
or, more generally,
\[
u(t)=z(t), \qquad z(t)=Fx(t).
\]
However, in many practical applications the full system state is
rarely directly measurable.
Limitations in sensing, communication constraints, distributed
architectures, and measurement delays often prevent direct access to
$x(t)$ and its delayed components.
Consequently, the implementation of stabilizing state-feedback
controllers requires the estimation of the control signal $z(t)$.
Functional observers provide a direct means of estimating the
desired functional $z(t)$ without reconstructing the entire
state vector.
This contrasts with state-observer-based schemes, where the full state
$x(t)$ must first be estimated before the functional $z(t)$ can be
computed.
Functional observers therefore offer a lower-dimensional estimation
mechanism that can simplify observer design and reduce computational
complexity.

The problem of functional estimation for time-delay systems has
received significant attention in the control literature \cite{dar2001}-\cite{Naami2021}.
To address the challenges introduced by delayed dynamics, various
functional observer structures and methodologies have been proposed.
Darouach in \cite{dar2001} established necessary and sufficient
conditions for the existence and design of functional observers for
continuous-time systems with delays in the state variables, where the
observer order is equal to the number of functionals to be estimated.
A systematic design procedure was also presented, and the stability of
the estimation error dynamics was analyzed using
Lyapunov--Krasovskii methods, leading to delay-dependent and
delay-independent stability conditions formulated as LMIs.
Subsequently, these results were extended to discrete-time systems in
\cite{dar2005}, where similar existence conditions and a parametric
observer design methodology were derived.
An alternative LMI-based sufficient condition, based on a different
delay-dependent functional observer structure, was later reported in
\cite{NamIJC2014}.
In \cite{Gu2022}, a parametric approach for designing functional
observers for linear time-varying systems with time delays was
presented, whereas in \cite{Naami2021} an $H_\infty$ functional
observer design method for nonlinear systems with multiple time
delays and disturbances was proposed. 
A notable feature of the above works is that the proposed functional
observer structures focus on the design of a specific-order observer,
typically equal to the number of functionals to be estimated, and
assume that instantaneous output measurements are available.
However, it is well known in functional observer theory that a
functional observer of this specific order may not always exist, in
which case the design methodology must allow for the construction of
higher-order observers \cite{ref10}.
Furthermore, in many practical situations the output measurements
themselves may be subject to delays, which further complicates the
observer design problem.
The present paper addresses this gap by developing a framework that
accommodates functional observers of different orders while explicitly
accounting for delays in the output measurements.

In this work, we explicitly distinguish between two types of delays.
The delay associated with the state evolution is denoted by $\tau$,
whereas the delay affecting the output measurements is denoted by $h$.
We address the problem of estimating a desired functional
$z(t)=Fx(t)$ for time-delay systems in which the state dynamics depend
on delayed states while the available output measurements are
themselves delayed.

Such constraints arise naturally in networked control systems, where
sensing, communication, and processing latencies prevent instantaneous
access to measured outputs for the purpose of functional state
reconstruction. The study of functional observers of different orders
under delayed output measurements is therefore of practical importance,
particularly in networked systems where functional observer techniques
have recently found increasing applications \cite{ref8n}--\cite{ref14n}.

Moreover, this paper highlights a fundamental structural distinction
between delay-free systems and time-delay systems. In the delay-free
case, the system state evolves in the $n$-dimensional space
$\mathbb{R}^n$, and functional observers are constructed as linear
mappings of the instantaneous state \cite{22new}. Even when augmentation
techniques are employed in the design of higher-order functional
observers (see \cite{ref10}), the resulting existence and synthesis
conditions remain confined to an $n$-dimensional state space and retain
a purely finite-dimensional algebraic character.

In contrast, time-delay systems exhibit an augmented state structure,
since their evolution depends not only on the instantaneous state
$x(t)$ but also on past state values. When $h \le \tau$, the plant
dynamics depend on the pair $\big(x(t),x(t-\tau)\big)$. When $h>\tau$,
the delayed state $x(t-h)$ does not appear in the plant dynamics,
but it may arise in the observer structure. Consequently, both
$x(t-\tau)$ and $x(t-h)$ may enter the observer dynamics.

Accordingly, the system can be represented in an augmented
finite-dimensional state space. If $h \le \tau$, the augmented state is
\[
\xi(t)=
\begin{pmatrix}
	x(t)\\
	x(t-\tau)
\end{pmatrix}
\in\mathbb{R}^{2n},
\]
whereas if $h>\tau$, it takes the form
\[
\xi(t)=
\begin{pmatrix}
	x(t)\\
	x(t-\tau)\\
	x(t-h)
\end{pmatrix}
\in\mathbb{R}^{3n}.
\]

This dimensional enlargement fundamentally distinguishes the delay
case from the delay-free setting. Consequently, functional observer
design must be formulated on this augmented delay-state space rather
than solely on the instantaneous state.

The introduction of \emph{generalized functionals} in this work is
motivated precisely by this structural enlargement. These functionals
act on both present and delayed state components, thereby providing
additional design flexibility and enabling the derivation of tractable
LMI-based existence conditions under delayed measurements.

\medskip
\noindent\textbf{Contributions:}
This paper introduces three functional observer structures and
establishes corresponding existence conditions together with
constructive synthesis procedures.
By employing a functional augmentation approach, verifiable algebraic
conditions are derived for observers of different orders.
Furthermore, the notion of \emph{generalized functionals} is
introduced by considering an augmented delayed state vector,
providing additional flexibility for satisfying functional observer
existence conditions and for systematic observer synthesis.

\medskip
\noindent\textbf{Organization:}
Section~II presents the system description and formulates the
problem.
Section~III introduces notation and reviews preliminary results used
throughout the paper.
Section~IV develops the existence conditions and observer synthesis
procedures for the case $h=\tau$, together with an illustrative numerical
example.
Section~V establishes the corresponding existence conditions and
synthesis procedure for the case $h>\tau$, also supported by a numerical
example.
Conclusions are drawn in Section~VI.
The stabilizability conditions for time-delay systems are formulated
as LMIs, and their derivation is provided in the Appendix.

\section{System Description and Problem Statement}\label{A}

Consider the linear time-delay system
\begin{IEEEeqnarray}{rcl}
	\label{1a}
	\dot{x}(t) &=& Ax(t) + A_\tau x(t-\tau) + Bu(t) \nonumber \\
	\label{1b}
	x(t) &=& \phi(t), \quad t\in[-h,0] \nonumber 
\end{IEEEeqnarray}
where $x(t)\in\mathbb{R}^n$ is the state vector, $u(t)\in\mathbb{R}^m$ is the control input vector,
$\tau>0$ is a constant state delay, and $h\ge\tau$ denotes the maximum delay appearing
in the available measurements. The initial function
$\phi:[-h,0]\to\mathbb{R}^n$ is given. The matrices
$A,A_\tau \in\mathbb{R}^{n\times n}$ and $B\in\mathbb{R}^{n\times m}$ are constant.

Due to sensing or communication constraints, the measured output vector is available
with delay and is modeled as
\begin{align}
	\label{2}
	y(t) = C_\tau x(t-\tau) + C_h x(t-h) \nonumber 
\end{align}
where $\tau>0$ and $h>0$ are constant time delays, $y(t)\in\mathbb{R}^{p}$, and
$C_\tau, C_h \in\mathbb{R}^{p\times n}$ are constant matrices.

However, if $0<h<\tau$, redefining the measured output as
$\tilde y(t)=y(t-\tau+h)$ yields an equivalent formulation in which
the smallest delay equals $\tau$.
In particular, when $h=\tau$, the output expression for $y(t)$ reduces to
\[
y(t) = Cx(t-\tau)
\]
where
\begin{equation}
C := C_\tau + C_h. \nonumber 
\end{equation}

The functional of interest is defined as
\begin{equation}\label{functional}
	z(t)=Fx(t) \nonumber 
\end{equation}
where $F\in\mathbb{R}^{r\times n}$ is full row rank.
The full row rank assumption ensures that the rows of $F$ are linearly independent,
so that the components of $z(t)$ represent independent linear functionals of the state.
This entails no loss of generality, since any linearly dependent rows can be removed
without altering the functional subspace spanned by $F$.

\medskip
We propose the following functional observer structures for the estimation of the desired functional.

\medskip
\noindent{\it Functional Observer Structure-A}

Structure~A is a delay-free internal dynamic observer driven by the delayed output measurements and input signals.
\begin{IEEEeqnarray}{rcl}
	\label{FOA1}
	\hat{z}(t)&=& w(t) + My(t) \nonumber \\
	\label{FOA2}
	\dot{w}(t) &=& Nw(t) + Gy(t) + G_\tau y(t-\tau)
	+ Ju(t)+ J_\tau u(t-\tau) \nonumber 
\end{IEEEeqnarray}

\noindent with initial condition
$w(0)$,
where $w(t)\in\mathbb{R}^r$ and $\hat{z}(t)$ denotes the estimate of $z(t)$.
The matrices $M$, $N$, $G$, $G_\tau$, $J$, and $J_\tau$ are to be determined such that
$\hat{z}(t)\to z(t)$ asymptotically as $t\to\infty$.

\medskip

\noindent{\it Functional Observer Structure-B}

Structure-B extends Structure-A by incorporating internal delay dynamics in the observer state.
\begin{IEEEeqnarray}{rcl}
	\label{FOB1}
	\hat{z}(t) &=& w(t) + My(t) \nonumber \\
	\label{FOB2}
	\dot{w}(t) &=& Nw(t) + N_\tau w(t-\tau)
	+ Gy(t) + G_\tau y(t-\tau) 
	+ Ju(t) \nonumber\\
	&&+ J_\tau u(t-\tau), \nonumber
\end{IEEEeqnarray}
with $w(\theta)=\rho(\theta)$  for $\theta\in[-\tau,0]$ where $w(t)\in\mathbb{R}^r$.
The matrices $M$, $N$, $N_\tau$, $G$, $G_\tau$, $J$, and $J_\tau$ are to be determined so that
$\hat{z}(t)\to z(t)$ asymptotically.

\medskip
\noindent{\it Functional Observer Structure-C}

Structure-C will be introduced subsequently as a further generalization,
incorporating additional internal delay channels to enhance design flexibility.
\begin{IEEEeqnarray}{rcl}
	\label{FOC1}
	\hat{z}(t)&=&w(t)+My(t)+M_\tau y(t-\tau)+M_hy(t-h) \nonumber \\
	\label{FOC2}
	\dot{w}(t) &=& Nw(t) + N_\tau w(t-\tau) + N_h w(t-h)
	+ Gy(t) \nonumber\\ 
	&&+ G_\tau y(t-\tau) +  G_h y(t-h) + G_{\tau\tau} y(t-2\tau) \nonumber\\
	&&+G_{\tau h}y(t-\tau -h)+G_{hh}y(t-2h)+ Ju(t) \nonumber\\ &&+ J_\tau u(t-\tau)+ J_h u(t-h)+J_{\tau\tau}u(t-2\tau) \nonumber\\ && + J_{\tau h} u(t-\tau-h)+J_{hh}u(t-2h) \nonumber 
\end{IEEEeqnarray}
with $w(\theta)=\rho(\theta)$ for $\theta\in[-h,0]$ where $w(t)\in\mathbb{R}^r$.
The matrices $M$, $M_\tau$, $M_h$, $N$, $N_\tau$, $N_h$, $G$, $G_\tau$, $G_h$, $G_{\tau\tau}$, $G_{\tau h}$, $G_{hh}$,  $J$, $J_\tau$, $J_h$, $J_{\tau\tau}$, $J_{\tau h}$ and $J_{hh}$ are to be determined so that
$\hat{z}(t)\to z(t)$ asymptotically.
\medskip

Each observer structure possesses distinct advantages and limitations in the
estimation of the desired functional, which will be analyzed in the sequel.

\section{Notation and Preliminaries}

For a matrix $G$, $G^{\mathsf T}$ denotes its transpose,
$G^{-}$ denotes a generalized inverse of $G$ satisfying
$GG^{-}G=G$, $\rank(G)$ denotes the rank of $G$, and
$\mathrm{sym}(\cdot)$ denotes the operator
\[
\mathrm{sym}(G):=G+G^{\mathsf T}.
\]
The symbols $\mathbb R$ and $\mathbb C$ denote the sets of real and complex numbers.
For $\lambda\in\mathbb C$, $\Re(\lambda)$ denotes its real part.

A matrix $G \in\mathbb{R}^{n\times n}$ is called positive definite 
($G>0$) if $G=G^{\mathsf T}$ and $x^{\mathsf T}Gx>0$ for all 
$x\neq \mathbf 0$.

For a square matrix $G$, $\sigma(G)$ denotes the set of all eigenvalues of $G$, and
\[
\sigma_{\min}(G)
=
\min\{\Re(\lambda):\lambda\in\sigma(G)\}.
\]

The identity matrix is denoted by $I$ when its dimension is clear 
from the context; otherwise it is specified using a subscript. 
The symbol $\mathbf 0$ denotes a zero matrix or zero vector; its 
dimension is inferred from the context and specified by a 
subscript only when necessary. The scalar zero is denoted by $0$.

For $G\in\mathbb{R}^{k\times n}$, we denote by 
$\operatorname{row}(G)\subseteq\mathbb{R}^n$ the subspace spanned 
by the rows of $G$. The direct sum of vector subspaces is denoted 
by $\oplus$.

The following lemmas will be used in the sequel.

\medskip

\begin{lemma}[see \cite{25}]\label{lem1}
	Let $\Theta\in\mathbb{R}^{k\times n}$ and $\Upsilon\in\mathbb{R}^{r\times n}$
	be given matrices, and consider the matrix equation
	\begin{equation}
		\label{9a}
		X\Theta = \Upsilon \nonumber 
	\end{equation}
	where $X\in\mathbb{R}^{r\times k}$ is unknown.
	
	A solution $X$ exists if and only if
	\begin{equation}
		\label{10}
		\rank
		\begin{pmatrix}
			\Upsilon \\ \Theta
		\end{pmatrix}
		=
		\rank(\Theta). \nonumber 
	\end{equation}
	In this case, the general solution is given by
	\begin{equation}
		\label{11ty}
		X = \Upsilon \Theta^{-}
		+ Z\big(I_k - \Theta \Theta^{-}\big) \nonumber 
	\end{equation}
	where 
	$Z\in\mathbb{R}^{r\times k}$ is an arbitrary matrix.
\end{lemma}
\medskip

\begin{remark}
	For dynamical systems without delays, the use of the general
	solution $X$ in Lemma~\ref{lem1} to design functional observers
	whose order equals the number of functionals to be estimated
	was first reported in \cite{22new}.
\end{remark}

\medskip

\begin{lemma}[see~\cite{22new}, \cite{refty1}]\label{lem1c}
	Let $\Psi$ be a matrix of full row rank and $\Lambda$ and $\Phi$
	be matrices with the same number of columns as $\Psi$ such that
	\[
	\rank\!\begin{pmatrix}\Lambda\\ \Psi\\ \Phi\end{pmatrix}
	=
	\rank\!\begin{pmatrix}\Psi\\ \Phi\end{pmatrix}.
	\]
	Define
	\[
	\ell := \rank(\Psi), \qquad 
	\mathcal I_1 := \begin{pmatrix} I_\ell\\ \bf0 \end{pmatrix},
	\]
	\[	N_1 := \Lambda
	\begin{pmatrix}\Psi\\ \Phi\end{pmatrix}^{-}\mathcal I_1,
	\qquad
	N_2 := \Big(I-
	\begin{pmatrix}\Psi\\ \Phi\end{pmatrix}
	\begin{pmatrix}\Psi\\ \Phi\end{pmatrix}^{-}\Big)\mathcal I_1 .
	\]
	Then, for every $\lambda\in\mathbb{C}$ with $\Re(\lambda)\ge 0$,
	\[
	\rank\!\begin{pmatrix}
		\lambda \Psi-\Lambda\\[1mm]
		\Phi
	\end{pmatrix}
	=
	\rank\!\begin{pmatrix}\Psi\\ \Phi\end{pmatrix}
	\]
	if and only if the pair $(N_1,N_2)$ is detectable, i.e.,
	\[
	\rank\!\begin{pmatrix}
		\lambda I - N_1\\
		N_2
	\end{pmatrix}
	= \ell
	\qquad \forall\,\lambda\in\mathbb{C},\ \Re(\lambda)\ge 0 .
	\]
\end{lemma}

\medskip
\begin{remark}
	Lemma 2 of \cite{22new} and also Lemma 9 of \cite{refty1} is exactly Lemma~\ref{lem1c} above where in \cite{22new}  and \cite{refty1}, $\Lambda, \Psi$ and $\Phi$ are chosen as $\Lambda= LA, \Psi = L$ and $\Phi = \begin{pmatrix}
		CA\\C
	\end{pmatrix}$.
\end{remark}

\medskip

\begin{lemma}\label{lem2}
		Consider the time-delay system
		\begin{IEEEeqnarray}{rcl} \label{A4ty}
			\dot e(t) &=& (N_1+KN_2)e(t)
			+ (N_{\tau_1}+KN_{\tau_2})e(t-\tau)
			\nonumber\\
			&&+ (N_{h_1}+KN_{h_2})e(t-h)
		\end{IEEEeqnarray}
		where $h>\tau>0$ are known constant delays,
		$N_1,N_{\tau_1},N_{h_1}\in\mathbb R^{n\times n}$,
		$N_2,N_{\tau_2},N_{h_2}\in\mathbb R^{m\times n}$ are known matrices,
		and $K\in\mathbb R^{n\times m}$ is to be designed.
	
	Suppose there exist a $3n\times 3n$-matrix $P>0$, four  $n\times n$-matrices 
	$Q_1> 0, Q_2> 0, R_1> 0, R_2> 0$, 
	a non-singular $n\times n $-matrix $M$, a $n\times m$-matrix $G$ and a scalar $\lambda$ such that the following LMI holds
	\begin{align}\label{eq:LMI}
		\Theta= &\mathrm{sym}(\Pi_1P\Pi_2^{\mathsf T}) +v_1Q_1v_1^{\mathsf T} -v_2(Q_1-Q_2)v_2^{\mathsf T} \nonumber\\
		&-v_3Q_2v_3^{\mathsf T}+\tau v_6 R_1v_6^{\mathsf T}+(h-\tau) v_6R_2v_6^{\mathsf T}\nonumber\\
		&-\frac{1}{\tau}\Big((v_1-v_2)R_1(v_1-v_2)^{\mathsf T}
		+3\rho_1R_1\rho_1^{\mathsf T}\Big)\nonumber\\
		&-\frac{1}{h-\tau}\Big((v_2-v_3)R_2(v_2-v_3)^{\mathsf T}
		+3\rho_2R_2\rho_2^{\mathsf T}\Big)\nonumber\\
		&+\mathrm{sym}((\lambda v_1+ v_6)M\mathcal{N}_1^{\mathsf T})+\mathrm{sym}((\lambda v_1+ v_6)G\mathcal{N}_2^{\mathsf T})\nonumber\\
		< &\ 0,
	\end{align}
	where
	\begin{align}
		&v_i:=\big(\textbf{0}_{n\times (i-1)n}\ \ I_n \ \ \textbf{0}_{n\times (6-i)n} \big)^{\mathsf T},\ i=1,2,\cdots, 6, \nonumber\\
		&\Pi_1:= \big(v_1\ \  \tau v_4 \ \ (h-\tau) v_5\big), \quad \Pi_2:=\big(v_6 \  v_1-v_2\  v_2-v_3\big), \nonumber\\		
		&\rho_1:=v_1+v_2-2v_4,  \quad \rho_2:= v_2+v_3-2v_5, \nonumber\\ 
		&\mathcal{N}_1^{\mathsf T}:=\big(N_1\quad N_{\tau_1}\quad N_{h_1}\quad \textbf{0}_{n\times 2n}\quad -I_n\big) \in \mathbb{R}^{n\times 6n},\nonumber\\
		&   \mathcal{N}_2^{\mathsf T}:=\big(N_2\quad N_{\tau_2}\quad N_{h_2}\quad \textbf{0}_{m\times 3n}\big)\in \mathbb{R}^{m\times 6n}.
		\nonumber
	\end{align}
	
	Then the system \eqref{A4ty} is asymptotically stable. 
	Moreover, a feasible choice of $K$ is given by
	\begin{equation}\label{eq:ZfromG}
		K=M^{-1}G. \nonumber 
	\end{equation}
\end{lemma}

\begin{proof}
	See Appendix.	
\end{proof}	

\bigskip

If the $h$-delay channel is absent, the error system reduces to a
single-delay system. In this case, the corresponding reduced
stability condition is given in the following lemma.

\medskip

\begin{lemma}\label{lem3}
	Consider the time-delay system
	\begin{IEEEeqnarray}{rcl}
		\dot e(t) &=& (N_1+KN_2)e(t)
		+ (N_{\tau_1}+KN_{\tau_2})e(t-\tau)
		\label{A4}
	\end{IEEEeqnarray}
	where $\tau>0$ is a known constant delay,
	$N_1,N_{\tau_1}\in\mathbb R^{n\times n}$,
	$N_2,N_{\tau_2}\in\mathbb R^{m\times n}$ are known matrices,
	and $K\in\mathbb R^{n\times m}$ is to be designed.
	
	Suppose there exist a $2n\times 2n$ matrix $\tilde P>0$,
	two $n\times n$ matrices $\tilde Q>0$, $\tilde R>0$,
	a non-singular matrix $\tilde M\in\mathbb R^{n\times n}$,
	a matrix $\tilde G\in\mathbb R^{n\times m}$,
	and a scalar $\lambda$ such that the following LMI holds:
	\begin{align}\label{eq:LMI2}
		\tilde{\Theta}
		=&\ \mathrm{sym}(\tilde{\Pi}_1\tilde{P}\tilde{\Pi}_2^{\mathsf T})
		+\tilde{v}_1\tilde{Q}\tilde{v}_1^{\mathsf T}
		-\tilde{v}_2\tilde{Q}\tilde{v}_2^{\mathsf T}
		+\tau \tilde{v}_4 \tilde{R}\tilde{v}_4^{\mathsf T} \nonumber\\
		&-\frac{1}{\tau}\Big(
		(\tilde{v}_1-\tilde{v}_2)\tilde{R}(\tilde{v}_1-\tilde{v}_2)^{\mathsf T}
		+3\tilde{\rho}\tilde{R}\tilde{\rho}^{\mathsf T}
		\Big)\nonumber\\
		&+\mathrm{sym}((\lambda \tilde{v}_1+\tilde{v}_4)\tilde{M}\tilde{\mathcal N}_1^{\mathsf T})
		+\mathrm{sym}((\lambda \tilde{v}_1+\tilde{v}_4)\tilde{G}\tilde{\mathcal N}_2^{\mathsf T})
		<0,
	\end{align}
	where
	\begin{align}
		&\tilde{v}_i
		:=
		\big(\mathbf 0_{n\times (i-1)n}\ \ I_n\ \ \mathbf 0_{n\times (4-i)n}\big)^{\mathsf T},
		\qquad i=1,\dots,4, \nonumber\\
		&\tilde{\Pi}_1:= \big(\tilde{v}_1\ \ \tau \tilde{v}_3\big),
		\qquad
		\tilde{\Pi}_2:= \big(\tilde{v}_4\ \ \tilde{v}_1-\tilde{v}_2\big), \nonumber\\
		&\tilde{\rho}:=\tilde{v}_1+\tilde{v}_2-2\tilde{v}_3, \nonumber\\
		&\tilde{\mathcal N}_1^{\mathsf T}
		:=
		\big(N_1\quad N_{\tau_1}\quad \mathbf 0_{n\times n}\quad -I_n\big)
		\in\mathbb R^{n\times 4n}, \nonumber\\
		&\tilde{\mathcal N}_2^{\mathsf T}
		:=
		\big(N_2\quad N_{\tau_2}\quad \mathbf 0_{m\times 2n}\big)
		\in\mathbb R^{m\times 4n}. \nonumber
	\end{align}
	
	Then the system \eqref{A4} is asymptotically stable.
	Moreover, a feasible stabilizing gain is given by
	\begin{equation}\label{eq:ZfromG2}
		K=\tilde{M}^{-1}\tilde{G}. \nonumber 
	\end{equation}
\end{lemma}

\begin{proof}
	See Appendix.	
\end{proof}	

\section{Existence and Design of Functional Observers for the Case $h=\tau$} \label{s1}
As established in the problem formulation, without loss of
generality it suffices to consider the case $h=\tau$ whenever
the delay associated with the output measurement does not
exceed the state delay. In this scenario, the measurement
equation reduces to
\[
y(t)=Cx(t-\tau),
\qquad C:=C_\tau+C_h,
\]
so that the plant and measurement delays are aligned.
This structural alignment permits a complete characterization
of the functional observer existence problem.

Consider the estimation of $z(t)=Fx(t)$ using
Functional Observer Structure-A when $h=\tau$.
Defining the estimation error $e(t)=\hat z(t)-z(t)$, the error dynamics are given by
\begin{align}
	\label{7ty1}
	\dot{e}(t)&=\dot{w}(t)+M\dot{y}(t)-F\dot{x}(t)\nonumber\\ \quad &=Ne(t)+\mathcal{L}_{1}x(t)+\mathcal{L}_{2}u(t) +\mathcal{L}_{3}u(t-\tau)+\mathcal{L}_{4}x(t-\tau)\nonumber\\ \quad &+\mathcal{L}_{5}x(t-2\tau)
\end{align}
where\\ 
$\mathcal{L}_1 = NF-FA,\quad \mathcal{L}_2 = J-FB, \quad\mathcal{L}_3 = J_\tau+MCB$\\
$\mathcal{L}_4 = GC-NMC+MCA-FA_\tau$, \quad $\mathcal{L}_5 = G_\tau C+MCA_\tau$ 
%
and 
$
\mathcal L
:=
\begin{pmatrix}
	\mathcal L_1 &
	\mathcal L_2 &
	\mathcal L_3 &
	\mathcal L_4 &
	\mathcal L_5
\end{pmatrix}.
$
\medskip

The following theorem characterizes the necessary and sufficient
conditions for the existence of a minimal-order Functional Observer
Structure-A. The minimal order is
\[
r := \rank(F),
\]
which equals the number of independent functionals to be estimated.

\medskip

\begin{thm}\label{thm:decoupleA}
	For $h=\tau$, Functional Observer Structure-A of order $r=\rank(F)$ provides
	asymptotic estimation of the functional $z(t)=Fx(t)$ and yields
	estimation error dynamics that are decoupled from the plant state
	if and only if $\mathcal L=\bf 0$ and $N$ is Hurwitz.
	In this case, the estimation error satisfies
	\[
	e(t)=\hat z(t)-z(t)\to {\bf 0} \quad \text{as} \quad t\to\infty
	\]
	for all admissible initial conditions and inputs $u(\cdot)$.
\end{thm}

\medskip

\begin{proof}
	Since $N\in\mathbb{R}^{r\times r}$, the internal dynamics of Functional Observer Structure-A are $r$-dimensional, implying that the observer order equals the dimension of the functional $z(t)\in\mathbb{R}^r$.
	
	\medskip
	
	\emph{Sufficiency:}
	If $\mathcal L=\bf 0$, then \eqref{7ty1} reduces to $\dot e(t)=Ne(t)$, so the error dynamics are
	decoupled from $x(\cdot)$. If, in addition, $N$ is Hurwitz, then $e(t)\to \bf 0$ as $t\to\infty$
	for all admissible initial conditions (independently of $u(\cdot)$).
	
	\emph{Necessity:}
	Assume the error dynamics are decoupled from the plant state and satisfy $e(t)\to \bf 0$ for all
	admissible initial conditions and all admissible inputs $u(\cdot)$.
	Decoupling from $x(\cdot)$ implies that the coefficients multiplying $x(t)$, $x(t-\tau)$,
	and $x(t-2\tau)$ in \eqref{7ty1} must vanish, hence $\mathcal L_1=\mathcal L_4=\mathcal L_5= \bf0$.
	Moreover, convergence for all inputs requires the terms $\mathcal L_2u(t)$ and $\mathcal L_3u(t-\tau)$
	to vanish for all admissible $u(\cdot)$, which implies $\mathcal L_2=\mathcal L_3=\bf0$.
	Therefore $\mathcal L=\bf0$, and \eqref{7ty1} reduces to $\dot e(t)=Ne(t)$.
Finally, asymptotic convergence
\[
e(t)\to {\bf 0} \quad \text{as} \quad t \to\infty
\]
for all initial conditions holds if and only if $N$ is Hurwitz.
This completes the proof.
\end{proof}

\medskip

\begin{remark}
	The decoupling condition $\mathcal L=\bf0$ ensures that the estimation error dynamics are autonomous, i.e.,
	\[
	\dot e(t)=Ne(t)
	\]
	and therefore independent of the plant state and input signals. This property is crucial because it guarantees that the convergence of the estimation error is determined solely by the internal observer dynamics and is unaffected by plant behaviour, external disturbances, or input excitations. In particular, once decoupling is achieved, the error convergence rate is governed exclusively by the eigenvalues of $N$.
\end{remark}

\medskip

We now characterize the solvability of the observer parameter equations
for the case $h=\tau$. As shown previously, asymptotic estimation with
error dynamics decoupled from the plant state requires the algebraic
condition $\mathcal L=\bf0$ together with the stability condition that $N$
be Hurwitz. The following theorem reformulates these requirements
as explicit rank and spectral conditions.

\medskip

\begin{thm}\label{thm:decouple_htau}
	Conditions $\mathcal L=\bf0$ and $N$ Hurwitz are equivalent to the following conditions
	\begin{itemize}
		\item[(i)] $\displaystyle
		\rank\!\begin{pmatrix}FA\\ F\end{pmatrix}= \rank(F)$,
		\item[(ii)] $FAF^{-}$ is Hurwitz,
		\item[(iii)] $\displaystyle
		\rank\!\begin{pmatrix}\Upsilon\\ \Theta \end{pmatrix}= \rank(\Theta)$,
	\end{itemize}
	where
	\begin{IEEEeqnarray}{rcl}
		\Theta &=&
		\begin{pmatrix}
			C  & \bf 0\\
			\bf 0  & C\\
			CA & CA_\tau
		\end{pmatrix},\qquad
		\Upsilon =
		\begin{pmatrix}
			FA_\tau & \bf 0
		\end{pmatrix}. \nonumber
	\end{IEEEeqnarray}
\end{thm}

\medskip 

\begin{proof}
	The condition $\mathcal L_1=\bf0$ is equivalent to
	\begin{equation}\label{eq21}
		NF=FA.
	\end{equation}
	By Lemma~\ref{lem1}, \eqref{eq21} admits a solution $N$ if and only if
	$\rank\!\begin{pmatrix}FA\\ F\end{pmatrix}=\rank(F)$, which gives item~(i).
	Moreover, since $F$ is full row rank there exists a right inverse $F^{-}$ with $FF^{-}=I_r$,
	and any solution of \eqref{eq21} satisfies $N=FAF^{-}$; hence $N$ is Hurwitz if and only if
	$FAF^{-}$ is Hurwitz, which yields item~(ii).
	
	Next, $\mathcal L_4=\bf0$ and $\mathcal L_5=\bf0$ are equivalent to the linear matrix equation
	\begin{equation}\label{eq22}
		X\Theta=\Upsilon,
	\end{equation}
	with $\Theta,\Upsilon$ as defined above and $X$ given by
		\begin{IEEEeqnarray}{rcl}
	X &=&
	\begin{pmatrix}
		\bar G & G_\tau & M
	\end{pmatrix},\qquad
	\bar G = G-NM.  \nonumber
		\end{IEEEeqnarray}
	By Lemma~\ref{lem1}, \eqref{eq22} is solvable
	if and only if $\rank\!\begin{pmatrix}\Upsilon\\ \Theta \end{pmatrix}=\rank(\Theta)$,
	which gives item~(iii). Moreover, the observer parameters contained in $X$ are obtained from the general solution of the matrix equation $X\Theta=\Upsilon$ given in Lemma~\ref{lem1}, where the free matrix $Z$ reflects the inherent design freedom available in the observer construction.

	Finally, $\mathcal L_2=\bf0$ and $\mathcal L_3=\bf0$ are satisfied by the direct choices
	$J=FB$ and $J_\tau=-MCB$, respectively. Therefore $\mathcal L=\bf0$ holds.
	
	Moreover, $\mathcal L=\bf0$ and $N$ Hurwitz are equivalent to
	items (i)–(iii), which completes the proof.
\end{proof}

\medskip

\begin{remark}
	When conditions (i)–(iii) are satisfied, the matrix $N$ is uniquely
	determined as
	\[
	N = FAF^{-}.
	\]
	The remaining observer parameters are obtained from the general
	solution of the linear matrix equation
	\[
	X\Theta = \Upsilon,
	\qquad
	X = \begin{pmatrix} \bar G & G_\tau & M \end{pmatrix},
	\]
	as given in Lemma~\ref{lem1}. 
	
	The free matrix $Z$ parameterizes the family of
	solutions of this equation, thereby determining the admissible
	triples $(\bar G, G_\tau, M)$. The observer gain $G$ is then recovered from
	\[
	G = \bar G + NM.
	\]
\end{remark}

\bigskip

The following corollary consolidates Theorems~\ref{thm:decoupleA}
and~\ref{thm:decouple_htau} into a single necessary and sufficient
existence condition for a minimal-order Functional Observer
Structure-A when $h=\tau$.

\medskip

\begin{corol}\label{cor1}
	For $h=\tau$, Functional Observer Structure-A of order $r = \rank(F)$ provides
	asymptotic estimation of the functional $z(t)=Fx(t)$ and yields
	estimation error dynamics decoupled from the plant state if and only if
	\begin{itemize}
		\item[(i)] $\displaystyle
		\rank\!\begin{pmatrix}FA\\ F\end{pmatrix}= \rank(F)$,
		\item[(ii)] $FAF^{-}$ is Hurwitz,
		\item[(iii)] $\displaystyle
		\rank\!\begin{pmatrix}\Upsilon\\ \Theta \end{pmatrix}= \rank(\Theta)$,
	\end{itemize}
	where
	\begin{IEEEeqnarray}{rcl}
		\Theta &=&
		\begin{pmatrix}
			C  & \bf0\\
			\bf0  & C\\
			CA & CA_\tau
		\end{pmatrix},\qquad
		\Upsilon =
		\begin{pmatrix}
			FA_\tau & \bf0
		\end{pmatrix}. \nonumber 
	\end{IEEEeqnarray}

\medskip

	\noindent In this case, the estimation error satisfies
	\[
	e(t)=\hat z(t)-z(t)\to \bf0 \quad \text{as \textit t}\to\infty
	\]
	for all admissible initial conditions and inputs $u(\cdot)$.
\end{corol}

\bigskip

Suppose that, for $h=\tau$, the existence conditions in Corollary~\ref{cor1}
for Functional Observer Structure-A of order $r$ are not satisfied.
In this case, a higher-order functional observer (i.e., of order greater
than $r$) may still exist.

The design of such a higher-order observer (if it exists) can be achieved by augmenting
the set of functionals to be estimated. Since the order of the observer
equals the number of functionals being estimated, increasing the observer
order corresponds to enlarging the functional subspace.

Let
\[
z_a(t)=Rx(t), \qquad R\in\mathbb{R}^{(q-r)\times n}, \quad q>r,
\]
be an additional functional such that
\[
\begin{pmatrix}
	F\\
	R
\end{pmatrix}
\]
has full row rank.
The augmented functional is then defined as
\[
z_{\mathrm{aug}}(t)
= \bar Fx(t) =
\begin{pmatrix}
	F\\
	R
\end{pmatrix}x(t),
\]
where $\bar F:=\begin{pmatrix}F\\ R\end{pmatrix}\in\mathbb{R}^{q\times n}$.

Consequently, the necessary and sufficient conditions for the existence 
of a Functional Observer Structure-A of order 
$q := \rank(\bar F)$, estimating the augmented functional 
$z_{\mathrm{aug}}(t)$ (and therefore estimate $z(t)$ as well), 
follow directly from Corollary~\ref{cor1} upon replacing 
$F$ by $\bar F$ and $r$ by $q$.

\medskip

\begin{corol}\label{cor2}
		For $h=\tau$, Functional Observer Structure-A of order $q = \rank(\bar F)$ provides
	asymptotic estimation of the functional 
	$z_{\mathrm{aug}}(t) := \bar Fx(t)
	 =\begin{pmatrix}F\\R\end{pmatrix}x(t)$ and yields
 estimation error dynamics decoupled from the plant state if and only if
 \begin{itemize}
 	\item[(i)] $\displaystyle
 	\rank\!\begin{pmatrix}\bar FA\\ \bar F\end{pmatrix}= \rank\begin{pmatrix}\bar F\end{pmatrix}$
 	\item[(ii)] $\bar FA\bar F^{-}$ is Hurwitz
 	\item[(iii)] $\displaystyle
 	\rank\!\begin{pmatrix}\bar \Upsilon\\ \Theta \end{pmatrix}= \rank(\Theta)$
 \end{itemize}
 where
 \begin{IEEEeqnarray}{rcl}
 	\Theta &=&
 	\begin{pmatrix}
 		C  & \bf0\\
 		\bf0  & C\\
 		CA & CA_\tau
 	\end{pmatrix},\qquad
 	\bar \Upsilon =
 	\begin{pmatrix}
 		\bar FA_\tau & \bf0
 	\end{pmatrix}. \nonumber 
 \end{IEEEeqnarray}
 In this case, the estimation error 
 \[
 e_{\mathrm{aug}}(t)=\hat z_{\mathrm{aug}}(t)-z_{\mathrm{aug}}(t)\to \bf0 \quad \text{as \textit t }\to\infty
 \]
 for all admissible initial conditions and inputs $u(\cdot)$.
\end{corol}	
\bigskip

\noindent{\bf Projector-based construction of $R$}

\medskip

Let
\[
\bar F_0 :=
\begin{pmatrix}
	F\\
	FA\\
	\vdots\\
	FA^{n-1}
\end{pmatrix},
\quad
q := \rank(\bar F_0),
\quad
r := \rank(F).
\]
Since $F$ is full row rank, define the orthogonal projector onto $\row(F)$ as
\[
P_F := F^{\mathsf T}(F F^{\mathsf T})^{-1}F.
\]
Define
\[
R_0 := (I - P_F)\bar F_0.
\]
Then
\[
\rank(R_0)=q-r.
\]
Let $R$ be any full--row--rank matrix whose rows form a basis of $\row(R_0)$.
Then $R\in\mathbb{R}^{(q-r)\times n}$ and
\[
\bar F :=
\begin{pmatrix}
	F\\
	R
\end{pmatrix}
\in\mathbb{R}^{q\times n}
\]
satisfies
\[
\row(\bar F)=\row(\bar F_0),
\]
and hence
\[
\rank\!\begin{pmatrix}
	\bar F A\\
	\bar F
\end{pmatrix}
=
\rank(\bar F).
\]

%

\medskip

With the above projector-based construction of $R$, the invariance condition
\[
\rank\!\begin{pmatrix}
	\bar F A\\
	\bar F
\end{pmatrix}
=
\rank(\bar F).
\]
is automatically satisfied. Consequently,  the existence conditions in Corollary~\ref{cor1} reduce to the remaining stability and solvability conditions stated below.

\bigskip

\begin{thm}\label{thm:aug-remove-i}
	Define the augmented functional
	\[
	z_{\mathrm{aug}}(t)=\bar F x(t)=
	\begin{pmatrix}z(t)\\ z_a(t)\end{pmatrix},
	\quad z(t)=Fx(t).
	\]
	For $h=\tau$, Functional Observer Structure-A of order $q=\rank(\bar F)$ provides asymptotic estimation of
	$z_{\mathrm{aug}}(t)$ (and hence estimation of $z(t)$ as well) and yields estimation error dynamics decoupled from the plant state
	if and only if
	\begin{itemize}
		\item[(i)] $\bar F A \bar F^{-}$ is Hurwitz,
		\item[(ii)] $\displaystyle
		\rank\!\begin{pmatrix}\bar\Upsilon\\ \Theta\end{pmatrix}=\rank(\Theta),
		$
	\end{itemize}
	where 
	\begin{IEEEeqnarray}{rcl}
		\Theta &=&
		\begin{pmatrix}
			C & \bf0\\
			\bf0 & C\\
			CA & CA_\tau
		\end{pmatrix},\qquad
		\bar\Upsilon =
		\begin{pmatrix}
			\bar F A_\tau & \bf0
		\end{pmatrix}.
	\end{IEEEeqnarray}
	In this case, the estimation error $e_{\mathrm{aug}}(t)=\hat z_{\mathrm{aug}}(t)-z_{\mathrm{aug}}(t)$ satisfies
	$e_{\mathrm{aug}}(t)\to \bf0$ as $t\to\infty$ (therefore $e(t)\to \bf0$ as $t\to\infty$ as well) for all admissible initial conditions and inputs $u(\cdot)$.
\end{thm}

\begin{proof}
	By construction, $\row(\bar F)=\row(\bar F_0)=\row\!\begin{pmatrix}F\\FA\\ \vdots\\FA^{n-1}\end{pmatrix}$.
	By the Cayley--Hamilton theorem, $FA^n$ is a linear combination of $F,FA,\dots,FA^{n-1}$, hence
	$\row(\bar F A)\subseteq \row(\bar F)$ and therefore
	\[
	\rank\!\begin{pmatrix}\bar F A\\ \bar F\end{pmatrix}=\rank(\bar F).
	\]
	Thus the invariance/rank condition (item (i) in Corollary~\ref{cor1} with $F$ replaced by $\bar F$)
	is automatically satisfied.
	
	Applying Corollary~\ref{cor1} to the augmented functional $z_{\mathrm{aug}}(t)=\bar F x(t)$
	(i.e., replacing $F$ by $\bar F$ throughout) yields that Structure-A of order $q$ achieves
	asymptotic estimation of $z_{\mathrm{aug}}(t)$ (and
	therefore asymptotic estimation of the original functional $z(t)$ as well) with decoupled error dynamics if and only if
	$\bar F A\bar F^{-}$ is Hurwitz and
	$\rank\!\begin{pmatrix}\bar\Upsilon\\ \Theta\end{pmatrix}=\rank(\Theta)$. This completes the proof.
\end{proof}

\medskip

Now suppose that, for $h=\tau$, the existence conditions
for Functional Observer Structure-A of order $r=\rank(F)$
according to Corollary~\ref{cor1}, and of order
$q=\rank(\bar F)$ according to Theorem~\ref{thm:aug-remove-i},
are not satisfied.

In this case, a higher–order functional observer (i.e., of order
greater than $q=\rank(\bar F)$) may still exist. Since the order of the observer
equals the number of functionals being estimated, increasing
the observer order corresponds to enlarging the admissible
functional class.
Specifically, we extend the instantaneous functionals
$\bar F x(t)$ to include {\it delay–dependent functionals}
\[
z_d(t)=F_d x(t-\tau),
\qquad
F_d \in \mathbb{R}^{\ell \times n}.
\]
Define the extended {\it delay state vector}
\[
\xi(t)
:=
\begin{pmatrix}
	x(t)\\
	x(t-\tau)
\end{pmatrix}
\in \mathbb{R}^{2n}.
\]
Let
\[
F_0 := 
\begin{pmatrix}
	\bar F\\
	\bf0
\end{pmatrix},
\qquad
F_\tau :=
\begin{pmatrix}
	\bf0\\
	F_d
\end{pmatrix},
\]
and define the extended functional compactly as
\[
z_{\mathrm{ext}}(t)
=
F_0 x(t) + F_\tau x(t-\tau)
=
\begin{pmatrix}
	F_0 & F_\tau
\end{pmatrix}
\xi(t).
\]
Denote
\[
F_{\mathrm{ext}}
:=
\begin{pmatrix}
	F_0 & F_\tau
\end{pmatrix}
\in \mathbb{R}^{s \times 2n},
\qquad
s := \rank(F_{\mathrm{ext}}).
\]
Thus, increasing the observer order corresponds to enlarging
the functional subspace from a subspace of $\mathbb{R}^n$
to a subspace of the {\it extended delay–state space}
$\mathbb{R}^n \oplus \mathbb{R}^n$.

\medskip

We now consider estimation of the extended functional
\[
z_{\mathrm{ext}}(t)=F_{\mathrm{ext}}\xi(t)
\]
using Functional Observer Structure-A when $h=\tau$.
Define the estimation error
\[
e_{\mathrm{ext}}(t)
=
\hat z_{\mathrm{ext}}(t)-z_{\mathrm{ext}}(t).
\]
The error dynamics can then be written as
\begin{IEEEeqnarray}{rcl}
	\dot e_{\mathrm{ext}}(t) &=& \dot{w}(t)+M\dot{y}(t)-F_0\dot{x}(t)-F_\tau\dot{x}(t-\tau) \nonumber\\
	&=&
	Ne_{\mathrm{ext}}(t)
	+\tilde{\mathcal L}_1 x(t)
	+\tilde{\mathcal L}_2 u(t)
	+\tilde{\mathcal L}_3 u(t-\tau)
	\nonumber\\
	&&
	+\tilde{\mathcal L}_4 x(t-\tau)
	+\tilde{\mathcal L}_5 x(t-2\tau)
	\label{eq:error_ext}
\end{IEEEeqnarray}
where
\medskip

\noindent $\tilde{\mathcal L}_1 = NF_0 - F_0A,\,\, \tilde{\mathcal L}_2 = J - F_0B, \,\, \tilde{\mathcal L}_3 = J_\tau - F_\tau B + MCB,$
$\tilde{\mathcal L}_4 = 	NF_\tau + GC - NMC + MCA - F_0A_\tau - F_\tau A,$\\
$\tilde{\mathcal L}_5 = G_\tau C + MCA_\tau - F_\tau A_\tau$\\
 and $\tilde{\mathcal L}
:=
\begin{pmatrix}
	\tilde{\mathcal L}_1 &
	\tilde{\mathcal L}_2 &
	\tilde{\mathcal L}_3 &
	\tilde{\mathcal L}_4 &
	\tilde{\mathcal L}_5
\end{pmatrix}.$
%
%

\medskip

We now present the following theorem characterizing the existence
of Functional Observer Structure-A of order
$s=\rank(F_{\mathrm{ext}})$.

\medskip

\begin{thm}\label{thm:xx_reduced}
	For $h=\tau$, Functional Observer Structure-A of order
	$s=\rank(F_{\mathrm{ext}})$ provides asymptotic estimation of
	$z_{\mathrm{ext}}(t)\in\mathbb{R}^{s}$ (and hence asymptotic estimation of $z(t)$ as well)
	with estimation error dynamics decoupled from the plant state
	if and only if the following conditions hold:
	\begin{itemize}
		\item[(i)]
		\[
		\rank
		\begin{pmatrix}
			\tilde\Upsilon\\
			\tilde\Theta
		\end{pmatrix}
		=
		\rank(\tilde\Theta)
		\]
		\item[(ii)] \begin{IEEEeqnarray}{rcl}
	&& \rank \begin{pmatrix}
		\lambda F_0 - F_0A & \lambda F_\tau - F_0A_\tau - F_\tau A &-F_\tau A_\tau \\
		\bf0 &CA &CA_\tau\\
		\bf0 &C &\bf0\\
		\bf0 &\bf0 &C
	\end{pmatrix} \nonumber \\
	&& = \rank \begin{pmatrix}
		F_0 &F_\tau &\bf0\\
		\bf0 &C &\bf0\\
		\bf0 &\bf0 &C\\
		\bf0 &CA &CA_\tau
	\end{pmatrix}, \quad \lambda\in\mathbb{C},\ \Re(\lambda)\ge 0 \nonumber 
\end{IEEEeqnarray}
equivalently 
the pair $(\tilde N_1,\tilde N_2)$ is detectable.
\end{itemize}
Here
\[
\tilde\Theta
=
\begin{pmatrix}
F_0 & F_\tau & \bf0\\
\bf0 & C & \bf0\\
\bf0 & \bf0 & C\\
\bf0 & CA & CA_\tau
\end{pmatrix},\]
\[\tilde\Upsilon
=
\begin{pmatrix}
F_0A & F_0A_\tau + F_\tau A & F_\tau A_\tau
\end{pmatrix},
\]
and
\[
\tilde N_1
=
\tilde\Upsilon\,\tilde\Theta^{-}\,\tilde{\mathcal I}_1,
\quad
\tilde N_2
=
\big(I_{s+3p}-\tilde\Theta\tilde\Theta^{-}\big)\tilde{\mathcal I}_1,
\]
where $\tilde{\mathcal I}_1$ denotes the first $s$ columns of $I_{s+3p}$.

	In this case, there exists a matrix $Z$ such that
	\[
	N=\tilde N_1+Z\tilde N_2
	\]
	is Hurwitz, and the estimation error satisfies
	\[
	e_{\mathrm{ext}}(t)\to \bf0
	\quad \text{as \textit t}\to\infty
	\]
	for all admissible initial conditions and inputs $u(\cdot)$.
\end{thm}

\medskip

\begin{proof}
	Following the same construction as in Theorem~\ref{thm:decoupleA}, 
	for $h=\tau$, Functional Observer Structure-A of order 
	$s=\rank(F_{\mathrm{ext}})$ provides asymptotic estimation of 
	$z_{\mathrm{ext}}(t)$ with error dynamics decoupled from the plant state 
	if and only if $\tilde{\mathcal L}=\bf0$ and the matrix $N$ is Hurwitz.
	
	The condition $\tilde{\mathcal L}=\bf0$ is equivalent to the linear matrix equation
	\begin{equation}\label{eq22ty}
		X\tilde\Theta=\tilde\Upsilon,
	\end{equation}
	where
	\begin{IEEEeqnarray}{rcl}
	X :=
	\begin{pmatrix}
		N & \bar G & G_\tau & M
	\end{pmatrix}
	\in\mathbb{R}^{s\times (s+3p)},
	\quad
	\bar G := G-NM . \nonumber 
	\end{IEEEeqnarray}
	
	By Lemma~\ref{lem1}, the matrix equation \eqref{eq22ty} is solvable 
	if and only if
	\[
	\rank\!\begin{pmatrix}\tilde\Upsilon\\ \tilde\Theta\end{pmatrix}
	=\rank(\tilde\Theta),
	\]
	which establishes item~(i).
	
	Whenever this condition holds, the general solution of \eqref{eq22ty} according to Lemma~\ref{lem1} is
	\begin{equation}\label{eq:X-general}
		X=\tilde\Upsilon\tilde\Theta^{-}
		+Z\big(I_{s+3p}-\tilde\Theta\tilde\Theta^{-}\big),
	\end{equation}
	where $Z\in\mathbb{R}^{s\times (s+3p)}$ is arbitrary.
	
	Let $\tilde{\mathcal I}_1$ denote the first $s$ columns of $I_{s+3p}$.
	Right-multiplying \eqref{eq:X-general} by $\tilde{\mathcal I}_1$ yields
	\[
	N = X\tilde{\mathcal I}_1
	= \tilde\Upsilon\tilde\Theta^{-}\tilde{\mathcal I}_1
	+ Z\big(I_{s+3p}-\tilde\Theta\tilde\Theta^{-}\big)\tilde{\mathcal I}_1
	=: \tilde N_1 + Z\tilde N_2.
	\]
	Hence, under $\tilde{\mathcal L}=\bf0$, the error dynamics reduce to
	\[
	\dot e_{\mathrm{ext}}(t)
	=
	\big(\tilde N_1+Z\tilde N_2\big)e_{\mathrm{ext}}(t).
	\]
	Therefore, asymptotic convergence
	$e_{\mathrm{ext}}(t)\to \bf0$ as $t\to\infty$
	is achievable if and only if there exists a matrix $Z$
	such that $\tilde N_1+Z\tilde N_2$ is Hurwitz.
	This is equivalent to detectability of the pair
	$(\tilde N_1,\tilde N_2)$. Now by letting
	\begin{IEEEeqnarray}{rcl}
		\Lambda &=& \begin{pmatrix}
			F_0A &F_0A_\tau+F_\tau A &F_\tau A_\tau
		\end{pmatrix} \nonumber\\
		\Psi &=& \begin{pmatrix}
			F_0 &F_\tau &\bf0
		\end{pmatrix} \nonumber\\
		\Phi &=& \begin{pmatrix}
			\bf0 &C &\bf0\\
			\bf0 &\bf0 &C\\
			\bf0 &CA &CA_\tau
		\end{pmatrix} \nonumber 
	\end{IEEEeqnarray}
	in Lemma~\ref{lem1c}, we obtain item (ii) which completes the proof.
\end{proof}

\medskip

\begin{remark}
	All observer parameters are embedded in
	\[
	X=\begin{pmatrix}N & \bar G & G_{\tau} & M\end{pmatrix}.
	\]
	Once condition~(i) is satisfied, the general solution of the
	decoupling equation \eqref{eq22ty} yields the parameterization
	\[
	N=\tilde N_1+Z\tilde N_2,
	\]
	where $Z$ is arbitrary.
	Selecting $Z$ such that $\tilde N_1+Z\tilde N_2$ is Hurwitz
	ensures asymptotic convergence of the estimation error $e_{\mathrm{ext}}(t)$.
	
	The matrices $N$, $\bar G$, $G_{\tau}$, and $M$
	are then obtained directly from the corresponding blocks of $X$,
	and the observer gains are recovered from
	\[
	G=\bar G+NM,
	\]
and \[J=F_0B, \ J_\tau=F_\tau B - MCB. \]	Hence the theorem provides a constructive and parameterized
	design procedure.
\end{remark}

Now suppose that, for $h=\tau$, none of the existence conditions derived for
Functional Observer Structure-A are satisfied, including the minimal-order case
of order $r=\rank(F)$ (Corollary~\ref{cor1}), the augmented instantaneous case
of order $q=\rank(\bar F)$ (Theorem~\ref{thm:aug-remove-i}), and the extended
delay-state case of order $s=\rank(F_{\mathrm{ext}})$
(Theorem~\ref{thm:xx_reduced}). Recall that Structure-A possesses no internal
delay dynamics in its observer state. In this situation, one may instead
consider a functional observer with internal delay dynamics, namely
Functional Observer Structure-B.

Recall the extended delay state vector
\[
\xi(t)
:=
\begin{pmatrix}
	x(t)\\
	x(t-\tau)
\end{pmatrix}
\in \mathbb{R}^{2n},
\]
and define the {\it generalized functional}
\[
h(t)
=
H_0 x(t) + H_\tau x(t-\tau)
=
\begin{pmatrix}
	H_0 & H_\tau
\end{pmatrix}
\xi(t).
\]
Since the order of a functional observer equals the dimension
of the functional being estimated, the previous cases can be
interpreted as special choices of $(H_0,H_\tau)$:

\begin{itemize}
	\item[(i)]
	If $H_0 = F$ and $H_\tau = \bf 0$, then $h(t)=z(t)$ and the observer
	order is $r=\rank(F)$.
	
	\item[(ii)]
	If $H_0 = \bar F=\begin{pmatrix}F\\R\end{pmatrix}$ and
	$H_\tau = \bf0$, then
	\[
	h(t)=\begin{pmatrix}z(t)\\ z_a(t)\end{pmatrix},
	\]
	and the observer order is $q=\rank(\bar F)$.
	
	\item[(iii)]
	If
	\[
	H_0 = 
	\begin{pmatrix}
		\bar F\\
		\bf0
	\end{pmatrix},
	\qquad
	H_\tau =
	\begin{pmatrix}
		\bf0\\
		F_d
	\end{pmatrix},
	\]
	then
	\[
	h(t)=\begin{pmatrix}
		z(t)\\
		z_a(t)\\
		z_d(t)
	\end{pmatrix},
	\]
	and the observer order is
	\[
	s=\rank\!\begin{pmatrix}H_0 & H_\tau\end{pmatrix}.
	\]
\end{itemize}

We now present a general analysis that characterizes the existence
of Functional Observer Structure-B for arbitrary choices of
$(H_0,H_\tau)$, thereby unifying the cases of orders $r$, $q$ and $s$.

Consider the estimation of \(
h(t)
=
H_0 x(t) + H_\tau x(t-\tau)\) using
Functional Observer Structure-B when $h=\tau$.
Defining the estimation error $e(t)=\hat h(t)-h(t)$, the error dynamics are given by
\begin{align}
	\label{7ty2}
	\dot{e}(t)&=\dot{w}(t)+M\dot{y}(t)-H_0\dot{x}(t) - H_\tau\dot{x}(t-\tau)\nonumber\\ \quad &=Ne(t)+N_{\tau}e(t-\tau)+\hat{\mathcal{L}}_{1}x(t)+\hat{\mathcal{L}}_{2}u(t) +\hat{\mathcal{L}}_{3}u(t-\tau)\nonumber\\ \quad &+\hat{\mathcal{L}}_{4}x(t-\tau)+\hat{\mathcal{L}}_{5}x(t-2\tau),
\end{align}
where
\medskip

\noindent $\hat{\mathcal{L}}_1 = NH_0-H_0A,\,\, \hat{\mathcal{L}}_2 = J-H_0B,\,\,\hat{\mathcal{L}}_3 = J_{\tau}-H_\tau B+MCB,$\\
$\hat{\mathcal{L}}_4 = NH_\tau + N_\tau H_0+GC-NMC+MCA-H_0A_\tau-H_\tau A,$\\
$\hat{\mathcal{L}}_5 = N_{\tau}H_\tau + G_\tau C-N_\tau MC+MCA_\tau-H_\tau A_\tau,$\\
and 
$
\hat{\mathcal L}
:=
\begin{pmatrix}
	\hat{\mathcal L}_1 &
	\hat{\mathcal L}_2 &
	\hat{\mathcal L}_3 &
	\hat{\mathcal L}_4 &
	\hat{\mathcal L}_5
\end{pmatrix}.$

\medskip

%

The following theorem characterizes a sufficient
condition for the existence of a Functional Observer Structure-B
of orders $r := \rank(F)$, $q=\rank(\bar F)$ and $s = \rank \begin{pmatrix}
	F_0 &F_\tau
\end{pmatrix}$ for specific coefficients $$\begin{pmatrix}
H_0 &H_\tau
\end{pmatrix} = \begin{pmatrix}
F &\bf0
\end{pmatrix}$$ 
and
$$\begin{pmatrix}
H_0 &H_\tau
\end{pmatrix} = \begin{pmatrix}
\bar F &\bf0
\end{pmatrix}$$ and $$\begin{pmatrix}
H_0 &H_\tau
\end{pmatrix} = \begin{pmatrix}
F_0 &F_\tau
\end{pmatrix} := \begin{pmatrix}
\bar F &\bf0\\\bf0 &F_d
\end{pmatrix} $$ respectively.

To derive the existence conditions for Functional Observer
Structure-B, we impose the decoupling condition
\[
\hat{\mathcal L}=\mathbf 0
\]
on the estimation-error dynamics. This ensures that the error
system is independent of the plant state and input.

Under this condition, the resulting constraints can be written as
the linear matrix equation
\begin{equation}\label{eq:XB}
	X\hat{\Theta}=\hat{\Upsilon},
\end{equation}
where
\begin{IEEEeqnarray}{rcl}
X &=&
\begin{pmatrix}
	N & N_\tau & \bar G & \bar G_\tau & M
\end{pmatrix},
\nonumber \\
\bar G &:=& G-NM,\nonumber \\
\bar G_\tau &:=& G_\tau-N_\tau M, \nonumber 
\end{IEEEeqnarray}
and
\begin{IEEEeqnarray}{rcl}
	\hat{\Theta}
	&=&
	\begin{pmatrix}
		H_0 & H_\tau & \mathbf 0\\
		\mathbf 0 & H_0 & H_\tau\\
		\mathbf 0 & C & \mathbf 0\\
		\mathbf 0 & \mathbf 0 & C\\
		\mathbf 0 & CA & CA_\tau
	\end{pmatrix},
	\nonumber\\
	\hat{\Upsilon}
	&=&
	\begin{pmatrix}
		H_0A & H_0A_\tau + H_\tau A & H_\tau A_\tau
	\end{pmatrix}.
	\nonumber
\end{IEEEeqnarray}

By Lemma~\ref{lem1}, the matrix equation \eqref{eq:XB} is solvable
if and only if
\[
\rank\!\begin{pmatrix}
	\hat{\Upsilon}\\
	\hat{\Theta}
\end{pmatrix}
=
\rank(\hat{\Theta}).
\]
Whenever this condition holds, the general solution is
\[
X
=
\hat{\Upsilon}\hat{\Theta}^{-}
+
Z\big(I_{2s+3p}-\hat{\Theta}\hat{\Theta}^{-}\big),
\]
where
\[
Z\in\mathbb R^{s\times(2s+3p)}
\]
is arbitrary.

Partitioning \(X\) yields
\[
N=\hat N_1+Z\hat N_2,
\qquad
N_\tau=\hat N_{\tau_1}+Z\hat N_{\tau_2},
\]
where
\[
\hat N_1
=
\hat{\Upsilon}\hat{\Theta}^{-}\hat{\mathcal I}_1,
\qquad
\hat N_2
=
\big(I_{2s+3p}-\hat{\Theta}\hat{\Theta}^{-}\big)\hat{\mathcal I}_1,
\]
\[
\hat N_{\tau_1}
=
\hat{\Upsilon}\hat{\Theta}^{-}\hat{\mathcal I}_2,
\qquad
\hat N_{\tau_2}
=
\big(I_{2s+3p}-\hat{\Theta}\hat{\Theta}^{-}\big)\hat{\mathcal I}_2,
\]
with
\[
\hat{\mathcal I}_1=
\begin{pmatrix}
	I_s\\
	\mathbf 0_{(s+3p)\times s}
\end{pmatrix},
\qquad
\hat{\mathcal I}_2=
\begin{pmatrix}
	\mathbf 0_{s\times s}\\
	I_s\\
	\mathbf 0_{3p\times s}
\end{pmatrix}.
\]

\medskip

\begin{thm}\label{thm:FOB_structB}
	Suppose that $h=\tau$. Consider Functional Observer
	Structure-B (with internal delay) of order
	\[
	s=\rank\!\begin{pmatrix}H_0 & H_\tau\end{pmatrix}.
	\]
	
	Then Structure-B achieves asymptotic estimation of
	$h(t)\in\mathbb R^s$ (and hence of $z(t)$),
	with estimation-error dynamics decoupled from the plant state,
	if the following conditions hold.
	
	\begin{itemize}
		
		\item[(i)]
		The rank condition
		\[
		\rank\!\begin{pmatrix}
			\hat{\Upsilon}\\
			\hat{\Theta}
		\end{pmatrix}
		=
		\rank(\hat{\Theta})
		\]
		is satisfied.
		
		\item[(ii)]
		There exists a matrix \(Z\) such that the delay-dependent error dynamics
		\[
		\dot e(t)
		=
		(\hat N_1+Z\hat N_2)e(t)
		+
		(\hat N_{\tau_1}+Z\hat N_{\tau_2})e(t-\tau)
		\]
		is asymptotically stable.
		
		A sufficient condition for this stability, according to
		Lemma~\ref{lem3}, is that the LMI \eqref{eq:LMI2} in
		Lemma~\ref{lem3} is feasible.
		
	\end{itemize}
	In this case the estimation error satisfies
	\[
	e(t)\to \mathbf 0
	\qquad \text{as } t\to\infty
	\]
	for all admissible initial conditions and inputs \(u(\cdot)\).
	
\end{thm}

\medskip

\begin{proof}
	Under the decoupling condition $\hat{\mathcal L}=\mathbf 0$, the
	error equation \eqref{7ty2} reduces to
	\[
	\dot e(t)=Ne(t)+N_\tau e(t-\tau),
	\]
	so that the estimation error dynamics are decoupled from the plant
	state $x(\cdot)$ and the input $u(\cdot)$.
	
	By Lemma~\ref{lem1}, the linear matrix equation
	\[
	X\hat{\Theta}=\hat{\Upsilon}
	\]
	is solvable if and only if
	\[
	\rank\!\begin{pmatrix}
		\hat{\Upsilon}\\
		\hat{\Theta}
	\end{pmatrix}
	=
	\rank(\hat{\Theta}),
	\]
	which establishes item~(i).
	
	Whenever this condition holds, the general solution yields
	\[
	N=\hat N_1+Z\hat N_2,
	\qquad
	N_\tau=\hat N_{\tau_1}+Z\hat N_{\tau_2},
	\]
	and hence the error dynamics takes the form
	\[
	\dot e(t)
	=
	\big(\hat N_1+Z\hat N_2\big)e(t)
	+
	\big(\hat N_{\tau_1}+Z\hat N_{\tau_2}\big)e(t-\tau).
	\]
	By Lemma~\ref{lem3}, a sufficient condition for asymptotic stability
	of this delay system is feasibility of the LMI in Lemma~\ref{lem3}.
	This establishes item~(ii), which completes the proof.
\end{proof}

\medskip

\begin{remark}[Observer parameter construction]
	Condition~(i) guarantees solvability of the decoupling equation
	\[
	X\hat\Theta=\hat\Upsilon,
	\]
	whose general solution is
	\[
	X
	=
	\hat{\Upsilon}\hat{\Theta}^{-}
	+
	Z\big(I_{2s+3p}-\hat{\Theta}\hat{\Theta}^{-}\big),
	\]
	where $Z\in\mathbb{R}^{s\times(2s+3p)}$ is arbitrary.
	
	Thus, all observer parameters are embedded in
	\[
	X=
	\begin{pmatrix}
		N & N_\tau & \bar G & \bar G_\tau & M
	\end{pmatrix}.
	\]
	For any choice of $Z$, the matrices
	\[
	N=\hat N_1+Z\hat N_2,
	\qquad
	N_\tau=\hat N_{\tau_1}+Z\hat N_{\tau_2}
	\]
	are obtained directly by projection.
	
	The remaining observer matrices are then recovered from the
	corresponding blocks of $X$, with the observer gains given by
	\[
	G=\bar G+NM, \ G_{\tau}=\bar G_{\tau}+N_{\tau}M,
	\]
	\[
	J=H_0B,\ J_{\tau}=H_\tau B -MCB.\]
	Consequently, the design of Functional Observer Structure-B
	reduces to selecting a matrix $Z$ such that the delay system
	\[
	\dot e(t)
	=
	\big(\hat N_1+Z\hat N_2\big)e(t)
	+
	\big(\hat N_{\tau_1}+Z\hat N_{\tau_2}\big)e(t-\tau)
	\]
	is asymptotically stable.
	According to Lemma~\ref{lem3}, this can be accomplished
	by solving the LMI condition in Lemma~\ref{lem3} with $N_1=\hat N_1$, $N_2=\hat N_2$, $N_{\tau_1}=\hat N_{\tau_1}$ and $N_{\tau_2}=\hat N_{\tau_2}$.
	Hence the proposed conditions yield a fully constructive
	parameterization and synthesis procedure.
\end{remark}

\medskip

We now present a numerical example to illustrate the proposed existence conditions and observer synthesis procedures described in this section.

\medskip

\textit{Example 1:}
Consider the time-delay system with system matrices
\[
A=
\begin{pmatrix}
	-2 & 1\\
	0 & -3
\end{pmatrix},\quad
A_\tau=
\begin{pmatrix}
	-4 & 1\\
	2 & 1
\end{pmatrix},\quad
B=
\begin{pmatrix}
	1\\
	2
\end{pmatrix}.
\]
The delayed output measurement is given by
\[
y(t)=Cx(t-\tau),
\]
where $C=I_2$.

For $\tau=1\,\text{s}$, it can be verified using the approach in
\cite{wu} that the time-delay system has unstable eigenvalues
located at $0.2104 \pm j\,2.3888$.
We consider the following six cases to illustrate the results of this section.

\medskip

\textit{Case 1:}
We design a functional observer using Structure-A to estimate the
functional
\[
z(t)=
\begin{pmatrix}
	0 & 1
\end{pmatrix}x(t)=x_2(t).
\]
With
\[
F=
\begin{pmatrix}
	0 & 1
\end{pmatrix},
\]
and the matrix $A$ given above, Condition~(i) of Corollary~\ref{cor1}
is satisfied.
Furthermore,
\[
N=FAF^{-}=-3,
\]
which is Hurwitz, and therefore Condition~(ii) of
Corollary~\ref{cor1} is satisfied.
Condition~(iii) can also be verified to hold.

For simplicity, letting $Z=\mathbf 0$ and applying Lemma~\ref{lem1}, the matrix
\[
X=
\begin{pmatrix}
	\bar G & G_{\tau} & M
\end{pmatrix}
\]
is obtained and 
consequently, the parameters of a first-order Functional Observer
Structure-A are determined as follows
\[
M=
\begin{pmatrix}
	-0.3061 & -0.4041
\end{pmatrix},\quad
N=-3,
\]
\[
G=
\begin{pmatrix}
	2.3061 & 1.3061
\end{pmatrix},\quad
G_{\tau}=
\begin{pmatrix}
	-0.4163 & 0.7102
\end{pmatrix},
\]
\[
J=2,\qquad J_{\tau}=1.1143.
\]
In this example, the functional observer predicts the current value
of $x_2(t)$ using delayed measurements.
Hence, the observer effectively compensates for the measurement delay
and may therefore be interpreted as a \textit{time-delay compensator}.

\medskip

\textit{Case 2:} We design a Functional Observer Structure-A to estimate the following functional
$$z(t)=\begin{pmatrix}
	2 &1
\end{pmatrix}x(t).$$
Now, with $F=\begin{pmatrix}
	2 &1
\end{pmatrix}$, Condition (i) of Corollary \ref{cor1}  is not satisfied since $\rank\!\begin{pmatrix}FA\\ F\end{pmatrix}=2$ but  $\rank(F)=1$. Thus, a first-order observer of Structure-A does not exist. As presented in Corrollary \ref{cor2} and Theorem \ref{thm:aug-remove-i}, we can now readily find an extra functional, $z_a(t)=Rx(t)$, where 
$$R=\begin{pmatrix}
	0 &1
\end{pmatrix},$$ to satisfy Condition (i) of Corollary \ref{cor2}, where $\bar{F}=\begin{pmatrix}
	F\\R
\end{pmatrix}.$  With $\bar{F}$ as obtained, Condition (ii) of Corollary \ref{cor2} is also satisfied where $N=\bar{F}A\bar{F}^-=A=\begin{pmatrix}
	-2 &1\\0 &-3
\end{pmatrix}$. Here, the eigenvalues of $N$ are the same as the eigenvalues of $A$, which are stable at $\{-2,-3\}$.

It is easy to verify that Condition (iii) of Corollary \ref{cor2} with $\bar \Upsilon =
\begin{pmatrix}
	\bar FA_\tau & \bf0
\end{pmatrix}$ is satisfied. By letting $Z=\mathbf 0$ and from from Lemma~\ref{lem1}, $X$ is obtained. Thus, we obtain the following observer parameters of a second-order Functional Observer Structure-A

\begin{IEEEeqnarray}{rcl}
&& M=\begin{pmatrix} 0.5510 &-0.2327\\
	-0.3061 &-0.4041\end{pmatrix}, \quad N=\begin{pmatrix}
	-2 &1\\0 &-3
\end{pmatrix}, \nonumber \\
&& G=\begin{pmatrix} -6.3061 &1.8122\\2.3061 &1.3061\end{pmatrix}, \quad G_{\tau}=\begin{pmatrix} 2.6694 &-0.3184\\-0.4163 &0.7102\end{pmatrix}, \nonumber \\
&& J=\begin{pmatrix}4\\2\end{pmatrix}, \quad J_{\tau}=\begin{pmatrix}-0.0857\\
	1.11434\end{pmatrix}. \nonumber 
\end{IEEEeqnarray}

\medskip

\textit{Case 3:} We reconsider Case 2 but now we consider a more stringent situation where the output measurement contains only one delayed state variable $x_1(t-\tau)$, i.e., $y(t)= \begin{pmatrix} 1 &0\end{pmatrix}x(t-\tau)$. In this scenario, to estimate the instantaneous functional $z(t)$ using Observer Structure-A with only delayed measurement $x_1(t-\tau)$, Condition (iii) of Corollary \ref{cor2} is not satisfied since $\rank(\Theta)=3$ but $\rank\!\begin{pmatrix}\bar{\Upsilon}\\ \Theta \end{pmatrix}=4$.

As presented in Theorem \ref{thm:xx_reduced}, we now estimate the extended functional $z_{ext}(t)=F_0x(t)+F_{\tau}x(t-\tau)$, $F_0=\begin{pmatrix}\bar{F}\\ \bf 0 \end{pmatrix}=\begin{pmatrix}2 &1\\0 &1\\0 &0 \end{pmatrix}$ and $F_{\tau}=\begin{pmatrix}\bf0\\ F_d \end{pmatrix}=\begin{pmatrix}0 &0\\0 &0\\a &b \end{pmatrix}$, $a$ and $b$ are scalars to be determined such that Condition (i) of Theorem \ref{thm:xx_reduced} is satisfied.

With $F_0$ and $F_{\tau}$ as given above, we obtain
$$\tilde{\Theta}=\begin{pmatrix} 2 &1 &0 &0 &0 &0\\0 &1 &0 &0 &0 &0
	\\0 &0 &a &b &0 &0\\
	0 &0 &1 &0 &0 &0\\
	0 &0 &0 &0 &1 &0\\
	0 &0 &-2 &1 &-4 &1
\end{pmatrix}.$$
From the above, $b\neq 0$ results in $\mathrm{rank}(\tilde{\Theta})=6$ and hence Condition (i) of Theorem \ref{thm:xx_reduced} is satisfied. For illustrative purpose, we select $a=0$,  $b=1$, and thus the extra delayed functional that we estimate is $x_2(t-1)$, which is independent of the output measurement, $y(t)=x_1(t-1)$.

Next, from (ii) of Theorem \ref{thm:xx_reduced}, we obtain $\tilde{N}_1=\begin{pmatrix} -2 &1 &3\\
	0 &-3 &1\\
	0 &0 &-4\end{pmatrix}$ and $\tilde{N}_2=\bf0$. The pair $(\tilde{N}_1, \tilde{N}_2)$ is detectable, and $N=\tilde{N}_1$ is Hurwitz.  Thus, Condition (ii) of Theorem \ref{thm:xx_reduced} is satisfied. We obtain the rest of the observer parameters as follows

$M=\begin{pmatrix} 0\\0\\1\end{pmatrix}$, $G=\begin{pmatrix}
	-3\\3\\-2\end{pmatrix}$, $G_{\tau}=\begin{pmatrix}
	0\\0\\6\end{pmatrix}$, $J=\begin{pmatrix}
	4\\2\\0\end{pmatrix}$ and $J_{\tau}=\begin{pmatrix}
	0\\0\\1 \end{pmatrix}.$

\medskip

\textit{Case 4:} We reconsider Case 1 but now the output measurement is the same as in Case 3, i.e., $y(t)= \begin{pmatrix} 1 &0\end{pmatrix}x(t-\tau)$.

However, Condition (iii) of Corollary \ref{cor1} is not satisfied since $\rank(\Theta)=3$ but $\rank\!\begin{pmatrix}\Upsilon\\ \Theta \end{pmatrix}=4$. In this situation, based on Theorem \ref{thm:xx_reduced}, we estimate the following extended functional $z_{ext}(t)=F_0x(t)+F_{\tau}x(t-\tau)$, where $F_0=\begin{pmatrix}F\\ \bf 0 \end{pmatrix}=\begin{pmatrix}0 &1\\0 &0 \end{pmatrix}$ and $F_{\tau}=\begin{pmatrix}\bf0\\ F_d \end{pmatrix}=\begin{pmatrix}0 &0\\0 &1 \end{pmatrix}$. With the extra delayed functional, Conditions (i)-(ii) of Theorem \ref{thm:xx_reduced} are satisfied. Thus, a second-order Observer Structure-A to estimate $z(t)=x_2(t)$ is obtained, where\\
$N=\begin{pmatrix} -3 &1\\0 &-4\end{pmatrix}$, $M=\begin{pmatrix} 0\\1\end{pmatrix}$, $G=\begin{pmatrix}
	3\\-2\end{pmatrix}$, $G_{\tau}=\begin{pmatrix}
	0\\6\end{pmatrix}$, $J=\begin{pmatrix}
	2\\0\end{pmatrix}$ and $J_{\tau}=\begin{pmatrix}
	0\\1 \end{pmatrix}.$

\medskip

\textit{Case 5:} We reconsider Case 1 but now matrix $A$ is given as follows
$$A=\begin{pmatrix}
	-2 &1\\0 &0.5
\end{pmatrix}.$$

Note that Observer Structure-A does not exist since $N=FAF^-=0.5$, which is not Hurwitz. Thus, we need to employ Observer Structure-B to ensure asymptotic stability of the designed observer.

According to Theorem \ref{thm:FOB_structB}, we design a first-order Observer Structure-B to estimate the functional $h(t)=z(t)=x_2(t)$ by letting $H_0 = F$ and $H_\tau = \bf0$.

Condition (i) of Theorem \ref{thm:FOB_structB} is satisfied, and we obtain $\hat N_1
=0.5$, $\hat N_2
=\bf 0$, $\hat N_{\tau_1}
= 0.6558$ and $\hat N_{\tau_2}\in\mathbb{R}^{8\times 1}\neq \bf0.$ 
Since $\hat{N}_{\tau_2}\neq \bf0$ and $\hat{N}_{\tau_2}(2,1)=0.5501\neq 0$,  we can choose $Z\in\mathbb{R}^{1\times 8}$ as follows
\[Z=\begin{pmatrix}
	0 &a &0 &0 &0 &0 &0 &0\end{pmatrix}.\] 
Thus, as per item (ii) of Theorem \ref{thm:FOB_structB}, we now need to ensure asymptotic stability of the following error time-delay system
\[
\dot e(t)
=0.5e(t)+N_{\tau}e(t-\tau), \ N_{\tau}=0.6558+0.5501a.
\]
Since the above is a scalar time-delay system, we can use the exact stability condition presented in \cite{mori}, i.e., $0.5+N_{\tau}<0$ and $N_{\tau}\geq -\frac{1}{\tau}$, to find $N_{\tau}$ to ensure asymptotic stability.

Let say, with $\tau=1s$, we pick $N_{\tau}=-0.65$, and hence we obtain $a=-2.3736$.

The error time-delay system $\dot{e}(t)=0.5e(t)-0.65e(t-1)$ is asymptotically stable, and based on \cite{wu}, we can compute its dominant eigenvalues at $-0.4538 \pm j0.3706$. Thus, Observer Structure-B allows the time-delay error dynamics
to be stabilized even when $N$ is unstable, a case for which
a functional observer cannot be realized using Structure-A.

Finally, a first-order Observer Structure-B is obtained, where

\medskip

\noindent $M=\begin{pmatrix} -0.0919 &0.0365\end{pmatrix}, N=0.5, N_{\tau}=-0.65,$\\
$G=\begin{pmatrix} 1.7703 &1.7419\end{pmatrix}, G_{\tau}=\begin{pmatrix}-0.3807 &0.0317\end{pmatrix},$\\
$J=2, J_{\tau}= 0.0190$.

\medskip

\textit{Case 6:} We reconsider Case 5 but now we consider the output measurement contains only one delayed state variable $x_1(t-\tau)$, i.e., $y(t)= \begin{pmatrix} 1 &0\end{pmatrix}x(t-\tau)$. 

In this scenario, a first-order Observer Structure-B to estimate the functional $h(t)=z(t)=x_2(t)$ by letting $H_0 = F$ and $H_\tau = \bf0$ does not exist. Here, we obtain $\hat N_1
=0.5$, $\hat N_2
=\bf 0$, $\hat N_{\tau_1}
= 1$ and $\hat N_{\tau_2}= \bf0$. Since  $\hat N_{\tau_2}= \bf0$, Condition (ii) of Theorem \ref{thm:FOB_structB} is not satisfied as the below error time-delay is not stable
\[\dot{e}(t)=0.5e(t)+e(t-\tau).\]

According to Theorem \ref{thm:FOB_structB}, we now design a second-order Observer Structure-B to estimate the functional $h(t)=\begin{pmatrix}z(t)\\ z_a(t)\end{pmatrix}$ by selecting $H_0 = \bar F=\begin{pmatrix}F\\R\end{pmatrix}=\begin{pmatrix}0 &1\\1 &0\end{pmatrix}$ and $H_\tau = \bf0$.

Conditions (i)-(ii) of Theorem \ref{thm:FOB_structB} are now satisfied. For $\tau=1s$, LMI (\ref{eq:LMI2}) of Lemma~\ref{lem3} is feasible and we obtain 
\[Z=\begin{pmatrix}
	0 &0 &0 &-8.5064 &0 &0 &0\\
	0 &0 &0 &-0.4496 &0 &0 &0\end{pmatrix},\] 
and the following error time-delay system is asymptotically stable
\[
\dot e(t)
=\begin{pmatrix}0.5 &0\\1 &-2\end{pmatrix}e(t)+\begin{pmatrix} 1 &-3.2532\\1 &-2.2248\end{pmatrix}e(t-1).
\]
Thus, we obtain a second-order Observer Structure-B, where

\medskip

\noindent $M=\bf 0$, $N=\begin{pmatrix}0.5 &0\\1 &-2\end{pmatrix}$, $N_{\tau}=\begin{pmatrix} 1 &-3.2532\\1 &-2.2248\end{pmatrix}$,\\ $G=\begin{pmatrix} 5.2532\\-1.7752\end{pmatrix}$, $G_{\tau}=\bf 0$, $J=\begin{pmatrix} 2\\1\end{pmatrix}$ and $J_{\tau}= \bf 0$.

\section{Existence and Design of Functional Observers for the Case $h>\tau$}
We now address the case $h>\tau$, where the measurement delay exceeds the
intrinsic state delay and the plant and measurement delays are no longer aligned.
Introduce the augmented delay state vector
\[
\breve \xi(t)
:=
\begin{pmatrix}
	x(t)\\
	x(t-\tau)\\
	x(t-h)
\end{pmatrix}
\in \mathbb{R}^{3n},
\]
and define the generalized functional
\begin{IEEEeqnarray}{rcl}
	\breve h(t)
	&=&
	H_0 x(t) + H_\tau x(t-\tau)+H_hx(t-h) \nonumber \\
	&=&
	\begin{pmatrix}
		H_0 & H_\tau & H_h
	\end{pmatrix}
	\breve \xi(t). \nonumber
\end{IEEEeqnarray}
We employ a more general Functional Observer Structure-C to estimate
$\breve h(t)$. Once $\breve h(t)$ is estimated, the desired functional
$z(t)=Fx(t)$ can be recovered through appropriate choices of
$\begin{pmatrix} H_0 & H_\tau & H_h \end{pmatrix}$.
Moreover, different observer orders arise from the following selections:
\begin{IEEEeqnarray}{rcl}
	r &:=& \rank(F)
	\ \text{if} \
	(H_0\; H_\tau\; H_h)=\begin{pmatrix}F &\bf0 &\bf0\end{pmatrix}
	\nonumber\\[1mm]
	q &:=& \rank(\bar F)
	\ \text{if} \
	(H_0\; H_\tau\; H_h)=\begin{pmatrix}
		F &\bf0 &\bf0\\R &\bf0 &\bf0\\
	\end{pmatrix}=\begin{pmatrix}\bar F &\bf0 &\bf0\end{pmatrix}
	\nonumber\\[1mm]
\breve s &:=& \rank(F_0\; F_\tau\; F_h)
	\ \text{if} \
	(H_0\; H_\tau\; H_h)=
	\begin{pmatrix}
		\bar F & \bf0 & \bf0\\
		\bf0 & F_\tau & \bf0\\
		\bf0 & \bf0 & F_h
	\end{pmatrix}.
	\nonumber
\end{IEEEeqnarray}
Defining the estimation error vector $e(t)=\breve h(t)-h(t)$, the error dynamics are given by
\begin{align}
	\label{s5ty2}
	\dot{e}(t) &=\dot{w}(t)+M\dot{y}(t)+M_{\tau}\dot{y}(t-\tau)+M_h\dot{y}(t-h)\nonumber\\&-H_0\dot{x}(t) - H_\tau\dot{x}(t-\tau)- H_h\dot{x}(t-h) \nonumber \\
	&=Ne(t)+N_\tau e(t-\tau)+ N_he(t-h)+ \breve{\mathcal{L}}_{1}u(t)\nonumber\\ &+\breve{\mathcal{L}}_{2}u(t-\tau) +\breve{\mathcal{L}}_{3}u(t-h)+\breve{\mathcal{L}}_{4}u(t-2\tau)\nonumber\\&+\breve{\mathcal{L}}_{5}u(t-2h)+\breve{\mathcal{L}}_{6}u(t-\tau-h)+\breve{\mathcal{L}}_{7}x(t)\nonumber\\&+\breve{\mathcal{L}}_{8}x(t-\tau)+\breve{\mathcal{L}}_{9}x(t-h) +\breve{\mathcal{L}}_{10}x(t-2\tau)\nonumber\\&+\breve{\mathcal{L}}_{11}x(t-2h)+\breve{\mathcal{L}}_{12}x(t-\tau-h)+\breve{\mathcal{L}}_{13}x(t-2\tau-h)\nonumber\\&+\breve{\mathcal{L}}_{14}x(t-\tau-2h)+\breve{\mathcal{L}}_{15}x(t-3\tau)+\breve{\mathcal{L}}_{16}x(t-3h).
\end{align}
where
\begin{IEEEeqnarray}{rcl}
&&\breve{\mathcal{L}}_1 =J-H_0B, \quad  \breve{\mathcal{L}}_2 =J_\tau-H_\tau B+MC_{\tau}B, \nonumber \\
&&	\breve{\mathcal{L}}_3 =J_h-H_hB+MC_hB, \quad \breve{\mathcal{L}}_4 =J_{\tau \tau}+M_{\tau}C_{\tau}B, \nonumber \\
&& \breve{\mathcal{L}}_5 =J_{hh}+M_hC_hB, \quad  \breve{\mathcal{L}}_6=J_{\tau h}+M_\tau C_hB+M_hC_{\tau}B, \nonumber\\
&&\breve{\mathcal{L}}_7 =NH_0-H_0A, \nonumber\\
&&\breve{\mathcal{L}}_8 =NH_\tau+N_\tau H_0+\bar{G}C_{\tau}+MC_{\tau}A-H_0A_\tau-H_\tau A,  \nonumber \\
&& \breve{\mathcal{L}}_9 =NH_h+N_hH_0+\bar{G}C_h+MC_hA-H_hA, \nonumber\\
&&\breve{\mathcal{L}}_{10} =N_\tau H_\tau +\breve{G}_{\tau}C_\tau +MC_{\tau}A_{\tau}+ M_{\tau}C_{\tau}A-H_{\tau}A_{\tau}, \nonumber \\
&& \breve{\mathcal{L}}_{11} =N_hH_h+\breve{G}_hC_h+M_hC_hA, \nonumber \\
&& \breve{\mathcal{L}}_{12}=N_\tau H_h+N_hH_\tau+\breve{G}_\tau C_h+\breve{G}_h C_{\tau}+MC_hA_{\tau} \nonumber \\
&& \qquad +M_\tau C_hA+M_hC_{\tau}A-H_hA_\tau, \nonumber\\
&&  \breve{\mathcal{L}}_{13}= \breve{G}_{\tau \tau}C_h+\breve{G}_{\tau h}C_{\tau}+M_\tau C_hA_{\tau}+M_hC_{\tau}A_{\tau}, \nonumber \\
&& \breve{\mathcal{L}}_{14}=\breve{G}_{h h}C_{\tau}+\breve{G}_{\tau h}C_h+M_hC_hA_\tau, \nonumber 
\end{IEEEeqnarray}

\[
\breve{\mathcal{L}}_{15}=\breve{G}_{\tau \tau}C_{\tau}+M_{\tau}C_{\tau}A_\tau, \quad  \breve{\mathcal{L}}_{16}=\breve{G}_{hh}C_{h}
\]

\noindent 
and

\medskip

\noindent $\breve{G}_\tau$, $\breve{G}_h$, $\breve{G}_{\tau \tau}$, $\breve{G}_{\tau h}$ and $\breve{G}_{hh}$ are defined as follows
\begin{IEEEeqnarray}{rcl}
&& \breve{G}_\tau = G_\tau-NM_\tau-N_\tau M, \quad \breve{G}_h = G_h-NM_h-N_hM, \nonumber \\
&& \breve{G}_{\tau \tau} = G_{\tau \tau}-N_\tau M_\tau, \quad \breve{G}_{\tau h} = G_{\tau h}-N_\tau M_h-N_hM_\tau,\nonumber\\
&& \breve{G}_{hh} = G_{hh}-N_hM_h. \nonumber 
\end{IEEEeqnarray}	
%
Moreover,
\[
\breve{\mathcal L}
:=
\begin{pmatrix}
	\breve{\mathcal L}_1 & \cdots &
	\breve{\mathcal L}_{16}
\end{pmatrix}.
\]

To derive the existence conditions for Functional Observer Structure-C,
we impose the decoupling condition $\breve{\mathcal L}=\mathbf 0$ on the
error dynamics. This condition ensures that the estimation error is
independent of the plant state and input signals. The resulting
decoupling constraints can be expressed as the linear matrix equation
\[
\breve X \breve{\Theta} = \breve{\Upsilon},
\]
where $\breve X$ collects the observer matrices to be determined.
The matrices $\breve{\Theta}$ and $\breve{\Upsilon}$ are defined as
follows.
\[
\breve{\Theta}:=
\begin{pmatrix}
	\breve \Theta_{12} &\mathbf 0\\
	\breve{\Theta}_{21} &\breve{\Theta}_{22}
\end{pmatrix},
\qquad
\breve{\Upsilon}:=
\begin{pmatrix}
	\breve{\Upsilon}_1 &\mathbf 0
\end{pmatrix},
\]
where
\[
\breve \Theta_{12}=
\begin{pmatrix}
	H_0 &H_{\tau} &H_h &\mathbf0 &\mathbf0 &\mathbf0\\
	\mathbf0 &H_0 &\mathbf0 &H_{\tau} &\mathbf0 &H_h\\
	\mathbf0 &\mathbf0 &H_0 &\mathbf0 &H_h &H_{\tau}\\
	\mathbf0 &C_{\tau} &C_h &\mathbf0 &\mathbf0 &\mathbf0\\
	\mathbf0 &\mathbf0 &\mathbf0 &C_{\tau} &\mathbf0 &C_h\\
	\mathbf0 &\mathbf0 &\mathbf0 &\mathbf0 &C_h &C_{\tau}
\end{pmatrix},
\]

\[
\breve \Theta_{21}=
\begin{pmatrix}
	\mathbf0 &\mathbf0 &\mathbf0 &\mathbf0 &\mathbf0 &\mathbf0\\
	\mathbf0 &\mathbf0 &\mathbf0 &\mathbf0 &\mathbf0 &\mathbf0\\
	\mathbf0 &\mathbf0 &\mathbf0 &\mathbf0 &\mathbf0 &\mathbf0\\
	\mathbf0 &C_{\tau}A &C_hA& C_{\tau}A_{\tau} &\mathbf0 &C_hA_{\tau}\\
	\mathbf0 &\mathbf0 &\mathbf0 &C_{\tau}A &0 &C_hA\\
	\mathbf0 &\mathbf0 &\mathbf0 &\mathbf0 &C_hA &C_{\tau}A
\end{pmatrix},
\]

\[
\breve \Theta_{22}=
\begin{pmatrix}
	C_h &\mathbf0 &C_{\tau} &\mathbf0\\
	C_{\tau} &C_h &\mathbf0 &\mathbf0\\
	\mathbf0 &C_{\tau} &\mathbf0 &C_h\\
	\mathbf0 &\mathbf0 &\mathbf0 &\mathbf0\\
	C_hA_{\tau} &\mathbf0 &C_{\tau}A_{\tau} &\mathbf0\\
	C_{\tau}A_{\tau} &C_hA_{\tau} &\mathbf0 &\mathbf0
\end{pmatrix},
\]

and
\[
\breve{\Upsilon}_1=
\begin{pmatrix}
H_0A&
H_0A_\tau+H_\tau A&
H_hA&
H_\tau A_\tau&
\mathbf0&
H_hA_\tau
\end{pmatrix}.
\]

By Lemma~\ref{lem1}, the linear matrix equation
$\breve X \breve{\Theta}=\breve{\Upsilon}$ is solvable if and only if
\[
\rank\!\begin{pmatrix}
	\breve{\Upsilon}\\
	\breve{\Theta}
\end{pmatrix}
=
\rank(\breve{\Theta})
\]
admits the general solution
\[
\breve X
=
\breve{\Upsilon}\breve{\Theta}^{-}
+
Z\big(I_{3\breve s+9p}-\breve{\Theta}\breve{\Theta}^{-}\big),
\]
where \(Z\) is arbitrary.
Partitioning \(\breve X\) yields
\[
N=\breve N_1+Z\breve N_2,\quad
N_\tau=\breve N_{\tau_1}+Z\breve N_{\tau_2},\quad
N_h=\breve N_{h_1}+Z\breve N_{h_2},
\]
with
\[
\breve N_1
=
\breve{\Upsilon}\breve{\Theta}^{-}\breve{\mathcal I}_1,
\quad
\breve N_2
=
\big(I_{3\breve s+9p}-\breve{\Theta}\breve{\Theta}^{-}\big)\breve{\mathcal I}_1,
\]

\[
\breve N_{\tau_1}
=
\breve{\Upsilon}\breve{\Theta}^{-}\breve{\mathcal I}_2,
\quad
\breve N_{\tau_2}
=
\big(I_{3\breve s+9p}-\breve{\Theta}\breve{\Theta}^{-}\big)\breve{\mathcal I}_2,
\]

\[
\breve N_{h_1}
=
\breve{\Upsilon}\breve{\Theta}^{-}\breve{\mathcal I}_3,
\quad
\breve N_{h_2}
=
\big(I_{3\breve s+9p}-\breve{\Theta}\breve{\Theta}^{-}\big)\breve{\mathcal I}_3.
\]

\medskip

\noindent The following theorem summarizes the existence conditions for
Functional Observer Structure-C.

\medskip

\begin{thm}\label{thm:6}
	Suppose that $h>\tau$. Consider Functional Observer Structure-C
	(with internal delays) of order
	\[
	\breve s=\rank\!\begin{pmatrix}H_0 & H_\tau & H_h\end{pmatrix}.
	\]
	
	Then Structure-C achieves asymptotic estimation of
	\begin{IEEEeqnarray}{rcl}
		\breve h(t)
		&=&
		H_0 x(t)+H_\tau x(t-\tau)+H_h x(t-h)
		\nonumber\\
		&=&
		\begin{pmatrix}
			H_0 & H_\tau & H_h
		\end{pmatrix}\bar\xi(t),
	\end{IEEEeqnarray}
	with estimation-error dynamics decoupled from the plant state,
	if the following conditions hold.
	
	\begin{itemize}
		
		\item[(i)]
		The rank condition
		\medskip
		
		\[
		\rank
		\begin{pmatrix}
			\breve{\Upsilon}\\
			\breve{\Theta}
		\end{pmatrix}
		=
		\rank(\breve{\Theta})
		\]
		is satisfied.
		
		\item[(ii)]
		There exists a matrix \(Z\) such that the error system
		\begin{IEEEeqnarray}{rcl}
			\dot e(t)
			&=&
			(\breve N_1+Z\breve N_2)e(t)
			+
			(\breve N_{\tau_1}+Z\breve N_{\tau_2})e(t-\tau)
			\nonumber\\
			&&+
			(\breve N_{h_1}+Z\breve N_{h_2})e(t-h)
		\end{IEEEeqnarray}
		is asymptotically stable.
		
		A sufficient condition for this stability,
		according to Lemma~\ref{lem2},
		is that the LMI in Lemma~\ref{lem2} is feasible.
		
	\end{itemize}
	
	In this case the estimation error satisfies
	\[
	e(t)\to \mathbf 0
	\quad \text{as } t\to\infty.
	\]
	
\end{thm}

\begin{proof}
	Under the decoupling condition
	\[
	\breve{\mathcal L}=\mathbf 0,
	\]
	the error equation \eqref{s5ty2} reduces to
	\[
	\dot e(t)=Ne(t)+N_\tau e(t-\tau)+N_h e(t-h),
	\]
	so that the estimation error dynamics are decoupled from the plant
	state \(x(\cdot)\) and the input \(u(\cdot)\).
	
	By construction, the condition \(\breve{\mathcal L}=\mathbf 0\) is
	equivalent to the linear matrix equation
	\[
	\breve X\breve{\Theta}=\breve{\Upsilon}.
	\]
	By Lemma~\ref{lem1}, this equation is solvable if and only if
	\[
	\rank\!\begin{pmatrix}
		\breve{\Upsilon}\\
		\breve{\Theta}
	\end{pmatrix}
	=
	\rank(\breve{\Theta}),
	\]
	which establishes item~(i).
	
	Whenever this rank condition holds, the general solution is
	\[
	\breve X
	=
	\breve{\Upsilon}\breve{\Theta}^{-}
	+
	Z\big(I_{3\breve s+9p}-\breve{\Theta}\breve{\Theta}^{-}\big),
	\]
	where \(Z\) is arbitrary. Partitioning \(\breve X\) then gives
	\[
	N=\breve N_1+Z\breve N_2,\quad
	N_\tau=\breve N_{\tau_1}+Z\breve N_{\tau_2},\quad
	N_h=\breve N_{h_1}+Z\breve N_{h_2}.
	\]
	Hence the decoupled estimation error dynamics takes the form
	\begin{IEEEeqnarray}{rcl}
		\dot e(t)
		&=&
		(\breve N_1+Z\breve N_2)e(t)
		+
		(\breve N_{\tau_1}+Z\breve N_{\tau_2})e(t-\tau)
		\nonumber\\
		&&+
		(\breve N_{h_1}+Z\breve N_{h_2})e(t-h). \nonumber 
	\end{IEEEeqnarray}
	
	Therefore, if there exists a matrix \(Z\) such that the above
	time-delay system is asymptotically stable, then Structure-C provides
	asymptotic estimation of \(\breve h(t)\).
	Moreover, by Lemma~\ref{lem2}, a sufficient condition for asymptotic
	stability of this delay system is feasibility of the LMI in
	Lemma~\ref{lem2}. This establishes item~(ii).
	
	Consequently, under conditions (i) and (ii), the estimation error
	satisfies
	\[
	e(t)\to \mathbf 0
	\qquad \text{as } t\to\infty.
	\]
	This completes the proof.
\end{proof}

\bigskip

We now present a numerical example to illustrate the proposed existence conditions and observer synthesis procedure described in this section.

\bigskip

\textit{Example 2:} This example illustrates the estimation of the desired functional
$z(t)$ using delayed measurements of the form
\[
y(t)=C_{\tau}x(t-\tau)+C_hx(t-h).
\]
Consider the following time-delay system
\[
\dot{x}(t)=2x(t)+x(t-\tau)+u(t).
\]

First, assume that the available measurement is
\[
y(t)=x(t-\tau),
\]
that is, the state variable $x(t)$ is measured with a delay $\tau$.
The objective is to estimate the instantaneous state $x(t)$.

We attempt to estimate $z(t)=x(t)$ using Functional Observer
Structure-B, which leads to the estimation-error time-delay system
\[
\dot{e}(t)=2e(t)+N_de(t-\tau),
\]
where the parameter $N_d$ must be chosen to ensure asymptotic
stability.

According to \cite{mori}, the exact stability condition for this
system is
\[
2+N_d<0, \qquad N_d\geq -\frac{1}{\tau}.
\]
These conditions imply that the delay $\tau$ must satisfy
\[
\tau<0.5\,\mathrm{s}.
\]
Therefore, when $\tau\geq0.5\,\mathrm{s}$, the estimation-error system
cannot be guaranteed to be asymptotically stable using
Functional Observer Structure-B.

Now, let us introduce the following delayed output measurement
\[y_e(t-\alpha)=x(t-h),\]
where $\alpha=h-\tau$. Thus, $y_e(t-\alpha)$ is obtained by delaying $y(t)$ by $\alpha$. We then consider the following new delayed measurement vector
\[
y_a(t)=\begin{pmatrix}
	y(t)\\y_e(t-\alpha)
\end{pmatrix}=C_{\tau}x(t-\tau)+C_hx(t-h),\]
where

\medskip

$C_{\tau}=\begin{pmatrix}
	1\\0\end{pmatrix}$ and $C_h=\begin{pmatrix}
	0\\1\end{pmatrix}$.

\medskip

Now consider using Functional Observer Structure-C to estimate
$z(t)=x(t)$. According to Theorem~\ref{thm:6}, both Conditions~(i)
and~(ii) are satisfied. Moreover, the LMI \eqref{eq:LMI} in
Lemma~\ref{lem2} is feasible for the delays $\tau=0.56\,\mathrm{s}$
and $h=1\,\mathrm{s}$. 

The resulting observer parameters are
\[
N=2, \qquad N_{\tau}=-2.3464, \qquad N_h=0.3441,
\]
and the corresponding estimation-error time-delay system
\[
\dot{e}(t)=2e(t)-2.3464e(t-0.56)+0.3441e(t-1)
\]
is asymptotically stable.

Consequently, Functional Observer Structure-C successfully estimates
the desired functional $z(t)=x(t)$ from the delayed output $y_a(t)$.
The observer parameters are given below.
\begin{IEEEeqnarray}{rcl}
	&&M= \begin{pmatrix}
		1.2549 &-0.1332\end{pmatrix}, M_{\tau}=\begin{pmatrix}
		-0.4183 &0.0222\end{pmatrix}, \nonumber \\ 
	&&M_h =\begin{pmatrix}
		0.0222 &0\end{pmatrix}, G=\begin{pmatrix}
		3.3464 &-0.3441\end{pmatrix}, \nonumber \\
	&&G_{\tau}=\begin{pmatrix}
		-4.1994 &0.3791\end{pmatrix}, G_h=\begin{pmatrix}
		0.4984 &-0.0458\end{pmatrix}, \nonumber \\
	&&G_{\tau \tau}=\begin{pmatrix}
		1.3998 &-0.0743\end{pmatrix}, G_{\tau h}=\begin{pmatrix}
		-0.2182 &0.0076\end{pmatrix} \nonumber\\
	&& G_{hh}=\begin{pmatrix}
		0.0076 &0\end{pmatrix}, J=1, J_{\tau}= -1.2549 \nonumber\\
	&&J_h=  0.1332, J_{\tau \tau}=0.4183, J_{\tau h}=-0.0444, J_{hh}=0. \nonumber 
\end{IEEEeqnarray}
	
\section{Conclusion}

This paper has investigated the problem of functional state estimation
for linear time-delay systems in which the delay associated with the
state evolution differs from the delay affecting the output measurements.
By explicitly distinguishing between the intrinsic state delay $\tau$
and the measurement delay $h$, a unified framework has been developed
for the estimation of functionals of the form $z(t)=Fx(t)$ under
mismatched delay conditions.

Three functional observer structures were introduced to accommodate
different delay configurations. For each structure, algebraic
existence conditions were established together with constructive
synthesis procedures. The proposed functional augmentation approach
enabled the derivation of verifiable rank-based conditions for observers
of various orders. In addition, the notion of generalized functionals,
defined over an augmented delayed state vector, provided increased
flexibility in satisfying observer existence conditions and in
systematic observer design.

For the cases $h=\tau$ and $h>\tau$, the resulting observer
design conditions were expressed in terms of solvable matrix
equations and, where appropriate, delay-dependent LMIs.
Numerical examples illustrated the applicability of the proposed
theory.
Future work may include extensions to systems with time-varying
delays, uncertain parameters, or nonlinear dynamics, as well as
integration with output-feedback stabilization frameworks for
networked control systems.

\section{Appendix: Stabilization of Time-Delay Error Systems}

\textit{Proof of Lemma \ref{lem2}:} Let 
\begin{align}
	\zeta_0^{\mathsf T}(t)
	&:=
	\left(
	e^{\mathsf T}(t)\quad
	\int_{t-\tau}^t e^{\mathsf T}(s)\,ds \quad
	\int_{t-h}^{t-\tau} e^{\mathsf T}(s)\,ds
	\right), \nonumber\\
	\xi^{\mathsf T}(t)
	&:=
	\Bigg(
	e^{\mathsf T}(t)\ \
	e^{\mathsf T}(t-\tau)\ \
	e^{\mathsf T}(t-h)\ \
	\frac{1}{\tau}\int_{t-\tau}^t e^{\mathsf T}(s)\,ds \nonumber\\
	&\qquad\qquad
	\frac{1}{h-\tau}\int_{t-h}^{t-\tau} e^{\mathsf T}(s)\,ds\ \
	\dot e^{\mathsf T}(t)
	\Bigg). \nonumber 
\end{align}

Consider the Lyapunov--Krasovskii functional
\begin{align}
	V(t,e_t) = &\zeta_0^{\mathsf T}(t)P\zeta_0(t) \nonumber \\
	&+\int_{t-\tau}^te^{\mathsf T}(s)Q_1e(s)ds
	+\int_{t-h}^{t-\tau}e^{\mathsf T}(s)Q_2e(s)ds \nonumber \\
	&+\int_{-\tau}^{0}\int_{\theta}^0\dot{e}^{\mathsf T}(t+u)R_1\dot{e}(t+u)dud\theta \nonumber \\
	& +\int_{-h}^{-\tau}\int_{\theta}^0\dot{e}^{\mathsf T}(t+u)R_2\dot{e}(t+u)du d\theta. \nonumber
\end{align}

It is easy to verify that
\begin{eqnarray}\label{eq:lambda_min}
	\lambda_{min}(P) ||e(t)||^2\leq V(t,e_t).\nonumber
\end{eqnarray}
Taking the derivative of $V$ in $t$, we have
\begin{align}
	\dot{V}(t,e_t)	=\ & 2\zeta_0^{\mathsf T}(t)P\dot{\zeta}_0(t)+ e^{\mathsf T}(t)Q_1e(t) \nonumber\\
	&-e^{\mathsf T}(t-\tau)(Q_1-Q_2)e(t-\tau)\nonumber\\
	&- e^{\mathsf T}(t-h)Q_2e(t-h)\nonumber\\
	&+\tau\dot{e}^{\mathsf T}(t)R_1\dot{e}(t)-\int_{t-\tau}^{t}\dot{e}^{\mathsf T}(s)R_1\dot{e}(s)ds
	\nonumber\\
	&+(h-\tau)\dot{e}^{\mathsf T}(t)R_2\dot{e}(t)-\int_{t-h}^{t-\tau}\dot{e}^{\mathsf T}(s)R_2\dot{e}(s)ds
	\nonumber\\
	=\ &\xi^{\mathsf T}(t)\Bigg(\mathrm{sym}(\Pi_1P\Pi_2^{\mathsf T}) +v_1Q_1v_1^{\mathsf T}\nonumber\\
	&-v_2(Q_1-Q_2)v_2^{\mathsf T}-v_3Q_2v_3^{\mathsf T}\nonumber\\
	&+v_6(\tau R_1+(h-\tau) R_2)v_6^{\mathsf T}\Bigg)\xi(t)\nonumber\\
	&-\int_{t-\tau}^{t}\dot{e}^{\mathsf T}(s)R_1\dot{e}(s)ds-\int_{t-h}^{t-\tau}\dot{e}^{\mathsf T}(s)R_2\dot{e}(s)ds. \label{eq:derivative-V}
\end{align}
By the Wirtinger-based integral inequality in \cite{Seuret2013}, we have
\begin{align}\label{eq:Wirtinger_ineq1}
	&-\int_{t-\tau}^{t}\dot{e}^{\mathsf T}(s)R_1\dot{e}(s)ds\nonumber\\
	&\quad \leq -\xi^{\mathsf T}(t)\frac{1}{\tau}\Big((v_1-v_2)R_1(v_1-v_2)^{\mathsf T}
	+3\rho_1R_1\rho_1^{\mathsf T}\Big)\xi(t)
\end{align}
and

\begin{align}\label{eq:Wirtinger_ineq2}
	&-\int_{t-h}^{t-\tau}\dot{e}^{\mathsf T}(s)R_2\dot{e}(s)ds\nonumber\\
	&\leq -\xi^{\mathsf T}(t)\frac{1}{h-\tau}\Big((v_2-v_3)R_2(v_2-v_3)^{\mathsf T}
	+3\rho_2R_2\rho_2^{\mathsf T}\Big)\xi(t).
\end{align}
On the other hand, by using the descriptor method \cite{Fridman2009}, we have
\begin{equation}\label{eq:descriptor1}
	0=
	2\xi^{\mathsf T}(t)(\lambda v_1+v_6)M(\mathcal N_1^{\mathsf T}+K\mathcal N_2^{\mathsf T})\xi(t).
\end{equation}
Putting $G=MK$. Then, \eqref{eq:descriptor1} follows that
\begin{eqnarray}\label{eq:descriptor2}
	0=2\xi^{\mathsf T}(t)(\lambda v_1+ v_6)(M\mathcal{N}_1^T+G \mathcal{N}_2^T)\xi(t).
\end{eqnarray}
From \eqref{eq:derivative-V} to \eqref{eq:descriptor2} and \eqref{eq:LMI}, we obtain 
\begin{eqnarray}\label{eq:upper_bound-deri_V}
	\dot{V}(t,v_t)\leq \xi^{\mathsf T}(t)\Theta\xi(t)<0.
\end{eqnarray}
By \eqref{eq:lambda_min} and \eqref{eq:upper_bound-deri_V}, we conclude that, with $K=M^{-1}G$, system \eqref{A4ty} is asymptotically stable. The proof of Lemma \ref{lem2} is completed. 

\bigskip

\textit{Proof of Lemma \ref{lem3}:} Let 

\begin{align}
	\zeta_0^{\mathsf T}(t)
	&:=
	\left(
	e^{\mathsf T}(t)\ \ 
	\int_{t-\tau}^t e^{\mathsf T}(s)ds
	\right)\nonumber\\
	\xi^{\mathsf T}(t)
	&:=
	\left(
	e^{\mathsf T}(t)\ \ 
	e^{\mathsf T}(t-\tau)\ \ 
	\frac{1}{\tau}\int_{t-\tau}^te^{\mathsf T}(s)ds\ \ 
	\dot e^{\mathsf T}(t)
	\right).
\end{align}

Consider the Lyapunov–Krasovskii functional

\begin{align}\label{eq:LKF1}
	V(t,e_t)
	=&\,
	\zeta_0^{\mathsf T}(t)\tilde{P}\zeta_0(t)
	+\int_{t-\tau}^te^{\mathsf T}(s)\tilde{Q}e(s)ds
	\nonumber \\
	&+\int_{-\tau}^{0}\int_{\theta}^0
	\dot{e}^{\mathsf T}(t+u)\tilde{R}\dot{e}(t+u)dud\theta . \nonumber 
\end{align}
Following the same derivation steps as in the proof of Lemma~\ref{lem2}, we also conclude the stability of system \eqref{A4} with $K=\tilde{M}^{-1}\tilde{G}$, and this completes the proof of Lemma \ref{lem3}.  
\newpage
\section*{Biography}
\begin{wrapfigure}{l}{1in}
	\includegraphics[width=1in,height=1.25in,clip,keepaspectratio]{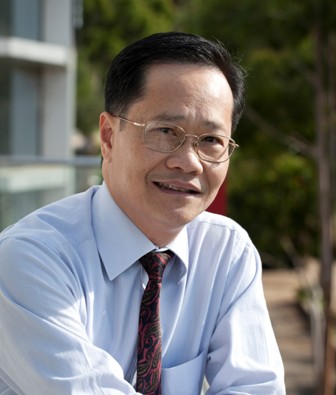}
\end{wrapfigure}
\noindent
\textbf{Hieu Trinh} received B.Eng. (Hons.), M.Eng.Sc., and Ph.D. degrees from the University of Melbourne, Australia, in 1990, 1992, and 1996, respectively. From 1997 to 2000, he was a Lecturer with the School of Engineering, James Cook University, Townsville, Australia. In 2001, he joined Deakin University, Geelong, VIC, Australia, where he is currently holds the position of Professor of Control Engineering, and Electrical Engineering. His research interests include theoretical control, observer design, and power system stability and control. 

\bigskip

\begin{wrapfigure}{l}{1in}
	\includegraphics[width=1in,height=1.25in,clip,keepaspectratio]{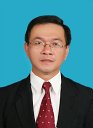}
\end{wrapfigure}
\noindent
\textbf{Phan T. Nam} received the B.Sc (Hons.), M.Sc. from Quy Nhon University, Diploma from ICTP, Italy, and Ph.D degree from Institute of Mathematics, Viet Nam, in 1997, 2002, 2003 and 2009, respectively. In 2003, He joined the  Department of Mathematics, Quy Nhon University (QNU), Viet Nam, where he currently holds the position of Associate Professor of Applied Mathematics. He is currently the Head of the Control Theory and Optimization 
Group at QNU.  His current research interests include 
stability analysis, state bounding, observer and controller design for (linear, positive, 2-D) time-delay systems.

\bigskip

\begin{wrapfigure}{l}{1in}
	\includegraphics[width=1in,height=1.25in,clip,keepaspectratio]{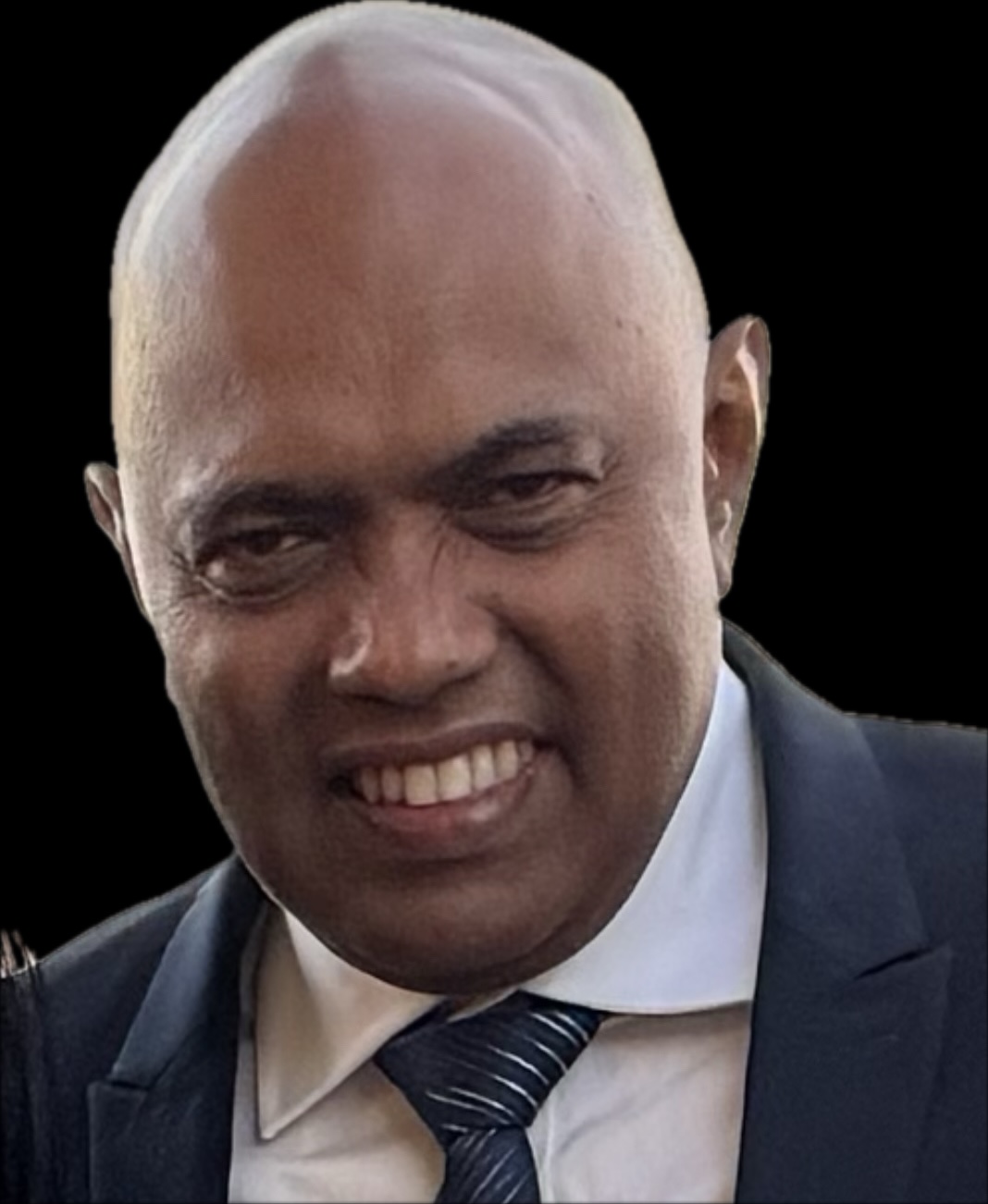}
\end{wrapfigure}
\noindent
\textbf{Tyrone Fernando} received his B.E. (Hons.) and Ph.D. degrees in Electrical Engineering from the University of Melbourne, Victoria, Australia, in 1990 and 1996, respectively. In 1996, he joined the Department of Electrical, Electronic and Computer Engineering, University of Western Australia (UWA), Crawley, WA, Australia, where he currently holds the position of Professor of Electrical Engineering. He previously served as Associate Head and Deputy Head of the department from 2008 to 2010.

Prof. Fernando is currently the Head of the Power and Clean Energy Research Group at UWA. He has provided professional consultancy to Western Power on the integration and management of distributed energy resources in the electric grid. In recognition of his professional contributions, he was named the Outstanding WA IEEE PES/PELS Engineer in 2018. His research interests include theoretical control, observer design, and power system stability and control. He has received multiple teaching awards from UWA in recognition of his contributions to control systems and power systems education.

\end{document}